\documentclass[aps, pra, reprint, superscriptaddress]{revtex4-1}
\usepackage{amsmath, amssymb}

\usepackage{graphicx}
\usepackage{dcolumn}
\usepackage{bm}
\usepackage{graphicx}
\usepackage{xcolor}
\usepackage{enumerate}
\usepackage[shortlabels]{enumitem}
\usepackage{amssymb}
\usepackage{amsthm}
\usepackage{amsmath,scalerel}

\usepackage[colorlinks=true]{hyperref}
\hypersetup
{
colorlinks,%
citecolor=green,%
linkcolor=blue,%
urlcolor=blue,%
}

\usepackage{braket}
\renewcommand{\v}[1]{\ensuremath{\mathbf{#1}}} 
\newcommand{\gv}[1]{\ensuremath{\text{\boldmath$ #1 $}}}
\newcommand{\abs}[1]{\left| #1 \right|} 
\newcommand{\norm}[1]{\left\| #1 \right\|} 

\newcommand{\eff}{\mathrm{eff}}
\newcommand{\expr}{\mathrm{exp}}
\let\baraccent=\= 
\renewcommand{\=}[1]{\stackrel{#1}{=}} 

\newcommand{\thmref}[1]{\hyperref[#1]{Theorem~\ref{#1}}}
\newcommand{\lemmaref}[1]{\hyperref[#1]{Lemma~\ref{#1}}}
\newcommand{\figref}[1]{\hyperref[#1]{Fig.~\ref{#1}}}
\newcommand{\figaref}[1]{\hyperref[#1]{Fig.~\ref{#1}a}}
\newcommand{\figbref}[1]{\hyperref[#1]{Fig.~\ref{#1}b}}
\newcommand{\figcref}[1]{\hyperref[#1]{Fig.~\ref{#1}c}}
\renewcommand{\eqref}[1]{\hyperref[#1]{Eq.~(\ref{#1})}}
\newcommand{\eqsref}[2]{\hyperref[#1]{Eq.~(\ref{#1})-(\ref{#2})}}
\newcommand{\appref}[1]{\hyperref[#1]{Appx.~(\ref{#1})}}
\DeclareMathAlphabet{\mathpzc}{OT1}{pzc}{m}{it}
\newtheorem{theorem}{Theorem}
\newtheorem{lemma}{Lemma}
\newcommand{\frakb}{\boldsymbol{\mathfrak{b}}}
\newcommand{\frakh}{\boldsymbol{\mathfrak{h}}}

\begin{document}

\title{Ancilla-free quantum error correction codes for quantum metrology}

\author{David Layden}
\thanks{These two authors contributed equally.}
\affiliation{Research Laboratory of Electronics and Department of Nuclear Science and Engineering, Massachusetts Institute of Technology, Cambridge, Massachusetts 02139, USA}

\author{Sisi Zhou}
\thanks{These two authors contributed equally.}
\affiliation{Departments of Applied Physics and Physics, Yale University, New Haven, Connecticut 06511, USA}
\affiliation{Yale Quantum Institute, Yale University, New Haven, Connecticut 06511, USA}

\author{Paola Cappellaro}
\affiliation{Research Laboratory of Electronics and Department of Nuclear Science and Engineering, Massachusetts Institute of Technology, Cambridge, Massachusetts 02139, USA}

\author{Liang Jiang}
\affiliation{Departments of Applied Physics and Physics, Yale University, New Haven, Connecticut 06511, USA}
\affiliation{Yale Quantum Institute, Yale University, New Haven, Connecticut 06511, USA}

\begin{abstract}
Quantum error correction has recently emerged as a tool to enhance quantum sensing under Markovian noise. It works by correcting errors in a sensor while letting a signal imprint on the logical state. This approach typically requires a specialized error-correcting code, as most existing codes correct away both the dominant errors and the signal. To date, however, few such specialized codes are known, among which most require noiseless, controllable ancillas. We show here that such ancillas are not needed when the signal Hamiltonian and the error operators commute; a common limiting type of decoherence in quantum sensors. We give a semidefinite program for finding optimal ancilla-free sensing codes in general, as well as closed-form codes for two common sensing scenarios: qubits undergoing dephasing, and a lossy bosonic mode. Finally, we analyze the sensitivity enhancement offered by the qubit code under arbitrary spatial noise correlations, beyond the ideal limit of orthogonal signal and noise operators.
\end{abstract}

\maketitle

\paragraph*{
Introduction.}
Quantum systems can make very effective sensors; they can achieve exceptional sensitivity to a number of physical quantities, among other features. However, as with most quantum technologies, the performance of quantum sensors is limited by decoherence. Typically, a quantum sensor acquires a signal as a relative phase between two states in coherent superposition \cite{giovannetti2006quantum,giovannetti2011advances,degen2017quantum}. Its sensitivity therefore depends both on how quickly this phase accumulates, and on how long the superposition remains coherent. The fundamental strategy to enhance sensitivity is then to increase the rate of signal acquisition (e.g., by exploiting entanglement) without reducing the coherence time by an equal amount \cite{huelga:1997}. These competing demands pose a familiar dilemma in quantum engineering: a quantum sensor must couple strongly to its environment without being rapidly decohered by it.

Quantum error correction (QEC) has recently emerged as a promising tool to this end. It is effective with DC signals and Markovian decoherence; important settings beyond the reach of dynamical decoupling, a widely-used tool with the same goal \cite{degen2017quantum,viola:1998, ban:1998}. The typical QEC sensing scheme involves preparing a superposition of logical states, and periodically performing a recovery operation (i.e., error detection and correction). This allows a signal to accumulate as a relative phase at the logical level, while also extending the duration of coherent sensing. For such a scheme to enhance sensitivity, however, great care must be taken in designing a QEC code which corrects the noise but not the signal. This new constraint is unique to error-corrected quantum sensing, and has no clear analog in quantum computing or quantum communication. Indeed, most QEC codes developed for these latter applications do not satisfy the above constraint, and so cannot be used for sensing. Device- and application-adapted QEC codes for sensing are therefore of timely relevance.

Recent works have begun to reveal how---and under what conditions---new QEC codes could enhance quantum sensing. Initial schemes assumed a signal and a noise source which coupled to a sensor in orthogonal directions, e.g., through $\sigma_z$ and $\sigma_x$ respectively \cite{kessler2014quantum,arrad2014increasing,unden2016quantum,dur2014improved,ozeri2013heisenberg, reiter:2017}. 
It was shown that unitary evolution could be restored asymptotically (in the sense that recoveries are performed with sufficiently high frequency) via a two-qubit code utilizing one probing qubit and one noiseless ancillary qubit \cite{kessler2014quantum,arrad2014increasing,unden2016quantum}. These results were generalized in Refs.~\cite{sekatski2017quantum,demkowicz2017adaptive,zhou2018achieving}, which showed that given access to noiseless ancillas, one can find a code that completely corrects errors without also correcting away the signal, provided the sensor's Hamiltonian is outside the so-called ``Lindblad span". (Intuitively, the Hamiltonian-not-in-Lindblad-span, or HNLS, condition means that the signal is not generated solely by the same Lindblad error operators one seeks to correct.) Ref.~\cite{layden2018spatial} then adapted this result to qubits with signal and noise in the same direction (say, both along $\sigma_z$), and found numerical evidence that noiseless ancillas were unnecessary in this common experimental scenario.

Noiseless, controllable ancillas have often been assumed for mathematical convenience in constructing QEC codes for sensing. While such ancillas are seldom available in experiment, little is known to date as to whether they are truly necessary for error-corrected quantum sensing, beyond limited counterexamples \cite{dur2014improved, ozeri2013heisenberg, reiter:2017, layden2018spatial}. Similarly, Refs.~\cite{zhou2018achieving,layden2018spatial} showed, through perturbative arguments, that QEC can still enhance sensitivity even when the HNLS condition is not exactly met. However, the exact sensitivity attainable with an error-corrected quantum sensor outside this ideal HNLS limit is unknown. We address both of these open questions in this work. First, we give a sufficient condition for error-corrected quantum sensing without noiseless ancillas, and a corresponding method to construct optimal QEC codes for sensing without ancillas. We then present new explicit codes for two archetypal settings: qubits undergoing dephasing, and a lossy bosonic mode. Finally, we introduce a QEC recovery adapted for the former code, and give an exact (i.e., non-perturbative) expression for the achievable sensitivity outside the HNLS limit.

\paragraph*{
QEC for sensing.}

We consider a $d$-dimensional sensor ($d < \infty$) under Markovian noise, whose dynamics is given by a Lindblad master equation \cite{breuer2002theory,gorini1976completely,lindblad1976generators}
\begin{equation}
\frac{d\rho}{dt} = -i[\omega H,\rho] + \sum_{i} \left(L_i \rho L_i^\dagger - \frac{1}{2}\{L_i^\dagger L_i,\rho\}\right),
\label{eq:lindblad}
\end{equation}
where $\omega H$ is the Hamiltonian from which $\omega$ is to be estimated, and $\{L_i\}$ are the Lindblad operators describing the noise. The Lindblad span associated with \eqref{eq:lindblad} is $\mathfrak{S} = {\rm span}\{I,L_i,L_i^\dagger,L_i^\dagger L_j,\,\forall i,j\}$, where ${\rm span}\{\cdot\}$ denotes the real linear subspace of Hermitian operators spanned by $\{\cdot\}$. 
One can use noiseless ancillas to construct a QEC code, described by the projector $P$, which asymptotically restores the unitary dynamics with non-vanishing signal
\begin{equation}
\label{eq:UnitaryChannel}
\frac{d\rho}{dt} = -i[\omega H_\eff,\rho], 
\end{equation}
where $H_\eff = P H P \not \propto P$, if and only if the HNLS condition is satisfied ($H \notin \mathfrak{S}$) \cite{zhou2018achieving}. To go beyond this result, we want to find conditions for QEC sensing codes that do not require noiseless ancillas, but still promise to reach the same optimal sensitivity. 
In parameter estimation, the quantum Fisher information (QFI) is used to quantify the sensitivity. According to the quantum Cram\'er-Rao bound~\cite{helstrom1976quantum,helstrom1968minimum,paris2009quantum,braunstein1994statistical}, the standard deviation $\delta \omega$ of the $\omega$-estimator is bounded by
$
\delta \omega \geq (N_\expr F(t))^{-1/2},
$
where $N_\expr$ is the number of experiments and $F(t)$ is the QFI as a function of the final quantum state. The bound is asymptotically achievable as $N_\expr$ goes to infinity~\cite{braunstein1994statistical,kobayashi2011probability,casella2002statistical}. For a pure state $\ket{\psi}$ evolving under Hamiltonian $\omega H$, $F(t) = 4t^2(\braket{\psi|H^2|\psi} - \braket{\psi|H|\psi}^2)$. Note that in this case $\delta \omega \propto 1/t$ is the so-called Heisenberg limit in time---the optimal scaling with respect to the probing time $t$ in quantum sensing \cite{giovannetti2006quantum,giovannetti2011advances,degen2017quantum}. In particular, the optimal asymptotic QFI provided by the error-corrected sensing protocol in Ref.~\cite{zhou2018achieving}, maximized over all possible QEC codes, is given by 
\begin{equation}
\label{eq:OptimalQFI}
F_{\rm opt}(t) = 4t^2 \min_{S \in \mathfrak{S}} \norm{ H - S }^2\equiv 4t^2 \norm{ H - \mathfrak S }^2,
\end{equation}
where $\norm{\cdot}$ is the operator norm.

\paragraph*{Commuting noise.}
We address here the following open questions: (i) Under what conditions the noiseless sensing dynamics in \eqref{eq:UnitaryChannel} can be achieved with an ancilla-free QEC code. (ii) Whether such code can achieve the same optimal asymptotic QFI in \eqref{eq:OptimalQFI} afforded by noiseless ancillas. We give a partial answer to these questions in terms of a sufficient condition on the signal Hamiltonian and the Lindblad jump operators in the following theorem.

\begin{theorem}[Commuting noise]
\label{thm:commuting}
Suppose $H \notin \mathfrak{S}$ and $[H, L_i] = [L_i, L_j] = 0$, $\forall i,j$. Then there exists a QEC sensing code without noiseless ancilla that recovers the Heisenberg limit in $t$ asymptotically. Moreover, it achieves the same optimal asymptotic QFI [\eqref{eq:OptimalQFI}] offered by noiseless ancillas.
\end{theorem}

\begin{proof}
A QEC sensing code  recovering \eqref{eq:UnitaryChannel} should satisfy the following three conditions \cite{zhou2018achieving}: 
\begin{gather}
P H P \not\propto P, \label{eq:QEC-1}\\
P L_i P \propto P, \quad
P L_i^\dagger L_j P \propto P,\label{eq:QEC-2}
\end{gather}
where $P = \ket{0_\textsc{l}} \!\! \bra{0_\textsc{l}} + \ket{1_\textsc{l}} \!\! \bra{1_\textsc{l}}$ is the projector onto the code space. \eqref{eq:QEC-2} is exactly the Knill-Laflamme condition to the lowest order in time evolution~\cite{bennett1996mixed,knill1997theory,nielsen2010quantum,beny2011perturbative} and \eqref{eq:QEC-1} is an additional requirement that the signal should not vanish in the code space. We say the code corrects the Lindblad span $\mathfrak S$ if \eqref{eq:QEC-2} satisfied. Without loss of generality, we consider only a 2-dimensional code $\ket{0_\textsc{l}} = \sum_{i=1}^d \sqrt{\beta^0_i} \ket{i}$, and $\ket{1_\textsc{l}} = \sum_{i=1}^d \sqrt{\beta^1_i} \ket{i}$, where $\{\ket{i}\}_{i=1}^d$ is an orthonormal basis under which $H$ and $L_i$'s are all diagonal. Define $d$-dimensional vectors $\gv 1, \gv h, \gv \ell_i$, and $\gv \ell_{ij}$ such that $(\v 1)_k = 1$, $(\gv h)_k = \braket{k|H|k}$, $(\gv{\ell_i})_k = \braket{k|L_i|k}$ and $(\gv{\ell_{ij}})_k = \braket{k|L^\dagger_i L_j|k}$. Define the real subspace $\mathfrak S_{\rm diag} = {\rm span}\{\v 1,{\rm Re}[\gv{\ell_i}],{\rm Im}[\gv{\ell_i}],{\rm Re}[\gv{\ell_{ij}}],{\rm Im}[\gv{\ell_{ij}}],\forall i,j\} \subseteq \mathbb R^{d}$. 
One can verify that the optimal code can be identified from the optimal solution $\gv\beta^* = \gv\beta^{0*} - \gv\beta^{1*}$ of the following semidefinite program (SDP) \cite{boyd2004convex}, where $\gv\beta^{0(1)}$ is the positive (negative) part of $\gv\beta$:
\begin{align}
\text{maximize}\,~~~& \braket{\gv\beta,\gv h}\label{eq:SDP-1}\\
\text{subject~to}~~~& \norm{\gv\beta}_1 \leq 2, \text{ and }\braket{\gv\beta,\gv \ell} = 0,\;\forall \gv \ell \in \mathfrak S_{\rm diag}.\label{eq:SDP-2}
\end{align}
Here $\norm{\gv x}_1 = \sum_{i=1}^d \abs{x_i}$ is the one-norm in $\mathbb{R}^d$ and $\braket{\gv x,\gv y} = \sum_{i=1}^d x_iy_i$ is the inner product.
Choosing the optimal input quantum state $\ket{\psi_{0}} = \frac{1}{\sqrt 2}(\ket{0_\textsc{l}}+\ket{1_\textsc{l}})$, the QFI at time $t$ is
$F(t) = t^2 \abs{\braket{\gv{\beta^0}-\gv{\beta^1},\gv h}}^2$.
Moreover, the optimal value of \eqref{eq:SDP-1} is $2 \min_{\gv \ell \in \mathfrak S_{\rm diag}} \norm{\gv h + \gv \ell}_\infty$ with the argument of the minimum denoted by $\gv\ell^\diamond$. Here $\norm{\cdot}_\infty$ denotes the infinity/max norm, defined by the largest absolute value of elements in a vector. The optimal solution $\gv\beta^{0(1)*}$ can be solved from the constraint that it is in the span of vectors $\gv v$ such that $\braket{\gv v,\gv h + \gv \ell^\diamond}$ is the largest (smallest) \cite{boyd2004convex}. In this case, $F(t) = 4t^2 \norm{\gv h - \mathfrak S_{\rm diag}}^2_\infty $ is the same as the optimal asymptotic QFI in \eqref{eq:OptimalQFI} that one can achieve with access to noiseless ancilla. Therefore, we conclude that $\gv\beta^{0(1)*}$ gives the optimal code. 
\end{proof}

\thmref{thm:commuting} reveals that the need for noiseless ancillas for QEC sensing arises from the non-commuting nature of the Hamiltonian and Lindblad operators. To this end, we give a non-trivial example with $[H, L_i] \neq 0$ for which there exist no ancilla-free QEC codes---even for arbitrarily large $d$---in \appref{app:necessary} of the Supplemental Material~\cite{SM}. Another interesting feature of commuting noise is that it allows quantum error correction to be performed with a lower frequency, by analyzing the evolution in the interaction picture~\cite{SM} (\appref{app:rotating}).  

We now consider two explicit, archetypal examples of quantum sensors dominated by commuting noise. In principle, a QEC code for each example could be found numerically through \thmref{thm:commuting}. Instead, however, we introduce two near-optimal, closed-form codes which are customized to the application and the errors at hand in both examples.

\paragraph*{Correlated Dephasing Noise.}

A common sensing scenario involves a quantum sensor composed of $N\ge 1$ probing qubits with energy gaps proportional to $\omega$ \cite{degen2017quantum}. For such a sensor to be effective, the qubits' energy gaps must depend strongly on $\omega$, which in turn makes them vulnerable to rapid dephasing due to fluctuations in their energies from a noisy environment. Consequently, dephasing is typically the limiting decoherence mode in such sensors \cite{biercuk:2009, witzel:2010, bluhm:2011, doherty:2013, muhonen:2014, orgiazzi:2016}. For simplicity, we suppose that each qubit has the same dephasing time $T_2$.
The generic Markovian dynamics for the sensor is then
\begin{equation}
\frac{d \rho}{dt}
=
-i [\omega H, \rho] + \frac{1}{2 T_2}
\sum_{j,k=1}^N
c_{jk} \Big( Z_j \rho Z_k - \frac{1}{2} \{ Z_j Z_k, \rho \} \Big).
\label{eq:dephasing_lindblad}
\end{equation}
Here, $H = \frac{1}{2} \gv \frakh \cdot \gv Z$ where $\gv Z = (Z_1, \dots, Z_N)$, so qubit $j$ has an energy gap $\omega \mathfrak h_j$. (Note that $\frakh \in \mathbb{R}^N$, whereas $\gv h$ in \eqref{eq:SDP-1} has dimension $d=2^N$.) The correlation matrix $C = (c_{jk})_{j,k=1}^N$
describes the spatial structure of the noise. It can be quite general, depending, say, on the proximity and orientation of the qubits to a nearby fluctuator, or on their coupling to a common resonator. In particular, $c_{jk} \in [-1,1]$ describes the correlation between the fluctuations on qubits $j$ and $k$, with the extremes $c_{jk}=1, -1 $ and 0 signifying full positive, full negative, and the absence of correlations, respectively. The assumption of identical $T_2$'s on all qubits is easily removed through a different, albeit less intuitive, definition of $C$ \cite{layden2018spatial}.

\eqref{eq:dephasing_lindblad} can be converted to the form of \eqref{eq:lindblad} by diagonalizing $C$. Concretely, $L_j = \sqrt{\lambda_j} \gv v_j \cdot \gv Z$ can be viewed as normal modes of the phase noise, where  $C \gv v_j = \lambda_j \gv v_j$ for some orthonormal eigenbasis $\{ \gv v_j \}_j \subset \mathbb{R}^N$. The HNLS condition then translates to $\frakh \notin {\rm col}(C)$, the column space of $C$, in this setting, which occurs when some normal mode overlapping with $H$ (i.e., $\gv v_u \cdot \frakh \neq 0$ for some $u \in [1,N]$) has a vanishing amplitude (i.e., $\lambda_u = 0$). This occurs generically in the limit of strong spatial noise correlations, provided the noise is not uniformly global \cite{layden2018spatial}. Observe that $[H, L_j]=[L_j, L_k]=0$ here, so \thmref{thm:commuting} guarantees a QEC code without noiseless ancillas saturating the optimal bound in \eqref{eq:OptimalQFI} when HNLS is satisfied. One such code, for $N \ge 3$, is given by
\begin{equation}
\label{eq:phase_code}
\begin{split}
\ket{0_\textsc{l}} &= \bigotimes_{j=1}^N
\Big[
\cos(\theta_j) \ket{0_j}
+
i \sin(\theta_j) \ket{1_j}
\Big]\\
\ket{1_\textsc{l}} &= X^{\otimes N} \ket{0_\textsc{l}},
\end{split}
\end{equation}
where $\gv\theta = \frac{1}{2}\arccos \frakb^\diamond$, defined element-wise, and $\frakb^\diamond$ is the solution of the following SDP: 
\begin{align}
\text{maximize}~~~& \braket{\frakb, \frakh}
\label{eq:qubit-SDP-1}\\
\text{subject~to}~~& \| \frakb \|_\infty \leq 1, \text{ and }\frakb \perp \mathrm{col}(C).
\label{eq:qubit-SDP-2}
\end{align}
It is straightforward to show that the code in  \eqref{eq:phase_code}, with this choice of $\frakb^\diamond$, satisfies the QEC sensing conditions Eqs.~(\ref{eq:QEC-1})--(\ref{eq:QEC-2}). It works by correcting all non-vanishing noise modes, but leaving a vanishing mode with the maximum overlap with $H$ uncorrected, through which $H$ affects the logical state. Moreover, it achieves the optimal asymptotic QFI in \eqref{eq:OptimalQFI}; in this case~\cite{SM} (\appref{app:dephasing_KL}): 
\begin{equation}
F_{\rm opt}(t) = t^2 \norm{\frakh - {\rm col}(C) }_1^2.
\label{eq:dephasing_QFI}
\end{equation}

Some remarks are in order. (i) Note that since the signal and noise here are both along $\sigma_z$ on each qubit, the usual repetition code \cite{repetition} is not suitable for sensing as it corrects not only the noise operators $L_j$, but also the signal Hamiltonian $H$. (ii) Remarkably, while the domain of the SDP in Eqs.~(\ref{eq:SDP-1})--(\ref{eq:SDP-2}) has dimension $O(2^N)$, that of Eqs.~(\ref{eq:qubit-SDP-1})--(\ref{eq:qubit-SDP-2}) only has dimension $O(N)$ due to our choice of ansatz. The ansatz in \eqref{eq:phase_code} therefore renders the QEC code optimization efficient. (iii) An approximate solution to Eqs.~(\ref{eq:qubit-SDP-1})--(\ref{eq:qubit-SDP-2}) is
\begin{equation}
\tilde{\frakb}^\diamond
=
\gamma \; \text{proj}_{\text{ker}(C)}  \frakh,
\label{eq:hC_tilde}
\end{equation}
where $\gamma$ is an adjustable parameter in the range $|\gamma| \in (0, \gamma_\text{max}]$, and $\gamma_\text{max} = \|\text{proj}_{\text{ker}(C)}  \gv \frakh \|_\infty^{-1}$. The code using $\gv\theta = \frac{1}{2}\arccos \tilde{\frakb}^\diamond$ always satisfies the QEC sensing conditions exactly [Eqs.~(\ref{eq:QEC-1})--(\ref{eq:QEC-2})], although it is approximate in that it need not saturate the optimal QFI in \eqref{eq:dephasing_QFI}. In the important case of a single vanishing noise mode [i.e., nullity$(C)=1$], however, \eqref{eq:hC_tilde} with $\gamma=\gamma_\text{max}$ achieves the optimal QFI.

\paragraph*{Lossy bosonic channel.}

Boson loss is often the dominant decoherence mechanism in a bosonic mode~\cite{chuang1997bosonic}, where the master equation is 
\begin{equation}
\frac{d\rho}{dt} = -i\Big[\sum_{i=1}^s \zeta_i (a^\dagger a)^i,\rho\Big] + \kappa \Big(a\rho a^\dagger - \frac{1}{2}\{a^\dagger a,\rho\}\Big),
\end{equation}
where $a$ is the annihilation operator, $\kappa$ is the boson loss rate and we only focus on the case where the Hamiltonian is a function of the boson number $a^\dagger a$. We apply a cutoff at the $s$-th power of $a^\dagger a$, where $s>1$ is a positive integer and ignore higher orders. We also truncate the boson number at a large number $M$, to make sure the system dimension is finite. According to the HNLS condition, $\zeta_1$ could not be sensed at the Heisenberg limit. However, it is possible to sense $\omega:=\zeta_s$ at the Heisenberg limit asymptotically, where the optimal code for the $s=2$ case is provided in Ref.~\cite{zhou2018achieving}. 

In order to sense $\omega$, it is important to filter out all lower-order signals $\sum_{i=1}^{s-1}\zeta_i(a^\dagger a)^{i}$ using the QEC code. Therefore, we should use the following modified Lindblad span~\cite{SM} (\appref{app:bosonic-noise}):
\begin{equation}
\label{eq:bosonic-lindblad}
\mathfrak S = \{I,a,a^{\dagger},(a^\dagger a)^i,\,1\leq i\leq s-1\}. 
\end{equation}
Note that the boson loss noise is not commuting because $[a,(a^\dagger a)^i] \neq 0$. However, this type of off-diagonal noise (if we use boson number eigenstates as basis) could be tackled simply by ensuring the distance of the supports (non-vanishing terms) of $\ket{0_{\textsc l}}$ and $\ket{1_{\textsc l}}$ is at least 3.

To obtain the optimal code, we could solve the SDP in Eqs.~(\ref{eq:SDP-1})--(\ref{eq:SDP-2}). However, when $M$ is sufficiently large, we could obtain a near-optimal solution analytically by observing that for large $M$, minimizing $\|(a^\dagger a)^s - \sum_{i=0}^{s-1} \chi_i (a^\dagger a)^i\|$ over all possible $\{\chi_i\}_{i=0}^{s-1}$ is equivalent to approximating a $s$-th degree polynomial using an $(s-1)$-degree polynomial.

The optimal polynomial is the Chebyshev polynomial \cite{chebyshev} and the near-optimal code is supported by its max/min points:
\begin{equation}
\label{eq:Chebyshev}
\begin{split}
\ket{0_\textsc{l}} &=\sum_{k \text{~even}}^{[0,s]} \tilde{c}_k \Ket{\left\lfloor M\sin^2\left( k\pi/{2s}\right) \right\rfloor},\\
\ket{1_\textsc{l}} &= \sum_{k \text{~odd}}^{[0,s]} \tilde{c}_k \Ket{\left\lfloor M\sin^2 \left(k\pi/{2s}\right) \right\rfloor},
\end{split}
\end{equation}
where $\lfloor x \rfloor$ means the largest integer smaller than or equal to $x$, $\abs{\tilde{c}_k}^2$ are positive numbers which can be obtained from solving a linear system of equations of size $O(s^2)$.  
$\abs{\tilde{c}_k}^2$ is approximately equal to $\frac{2}{s} - \frac{1}{s} \delta_{ks} - \frac{1}{s} \delta_{k0}$ for sufficiently large $M$. It is interesting to note that the supports of $\ket{0_{\textsc l}}$ and $\ket{1_{\textsc l}}$ bears a resemblance to the optimal time intervals in Uhrig dynamical decoupling~\cite{uhrig2007keeping}. Detailed calculations are provided in Ref.~\cite{SM} (\appref{app:bosonic-code}). 

We call \eqref{eq:Chebyshev} the $s$\textit{-th order Chebyshev code}. From the point of view of  quantum memories, the $s$-th order Chebyshev code could correct $s-1$ dephasing events, $L$ boson losses and $G$ gains, when $L+G=\left\lfloor\frac{M}{2}\sin^2\left(\frac{\pi}{s}\right)-1\right\rfloor$~\cite{michael2016new}. In terms of error-corrected quantum sensing, it corrects the Lindblad span (\eqref{eq:bosonic-lindblad}) and provides a near optimal asymptotic QFI for $\omega$
\begin{equation}
F(t) \approx F_{\rm opt}(t) \approx 16t^2\left(\frac{M}{4}\right)^{2s},
\end{equation}
for sufficiently large $M$. Note that the $(s-1,\frac{M}{s}-1)$ binomial code~\cite{michael2016new} also corrects~\eqref{eq:bosonic-lindblad}, but it gives a QFI that is exponentially smaller than the optimal value by a factor of $O\big(s(2/e)^{2s}\big)$ for a sufficiently large $M$.

\paragraph*{Enhancing sensitivity beyond HNLS.}

Previous works have focused on regimes where the HNLS condition is exactly satisfied. This is the ideal scenario where QEC can, in principle, suppress decoherence arbitrarily well while maintaining the signal, thus achieving the Heisenberg limit in time. However, QEC can still enhance quantum sensing well beyond this ideal case. The main difference is that when the HNLS condition is not satisfied, the encoded dynamics of the sensor will not be completely unitary (even asymptotically for $\Delta t \rightarrow 0$), in contrast with \eqref{eq:UnitaryChannel}. Yet, decoherence at the logical level can often be made weaker than that at the physical level---while still maintaining signal---giving a net enhancement in sensitivity. 

To show how, we generalize the example of $N \ge 3$ qubits under phase noise to this more realistic setting. When HNLS is satisfied, the code in \eqref{eq:phase_code} corrects noise modes with non-zero amplitude $\lambda_j > 0$, but leaves a mode with $\lambda_u=0$ uncorrected. The signal Hamiltonian is then detected through its projection onto the uncorrected noise mode $\sim \gv v_u \cdot \gv Z$. In experiment, however, the noise correlation matrix $C$ is generically full-rank, meaning that the HNLS condition is not satisfied, i.e., $C$ has no vanishing eigenvalues. Yet, non-trivial noise correlations (i.e., $C \not \propto I$) will generally cause $C$ to have a non-uniform spectrum, meaning that some eigenvalues, and therefore some $L_j$'s, will be subdominant. It is therefore possible to adapt the HNLS approach to this scenario by designing a code that accumulates signal at the cost of leaving uncorrected just one subdominant noise mode ($\lambda_u \approx 0$). This is done through an appropriate choice of $\gv \theta$ in \eqref{eq:phase_code}. To reach a closed-form expression for the resulting sensitivity, we use \eqref{eq:hC_tilde} as a starting point rather than an SDP formulation. Concretely, since $\text{ker}(C)=\{ \gv 0 \}$ for generic $C$, we take
\begin{equation}
    \gv \theta = \frac{1}{2} \arccos(\gamma \gv v_u),
\label{eq:theta_nonHNLS}
\end{equation}
defined element-wise, where $\gamma \in (0, \gamma_\text{max}]$ is again adjustable, now with $\gamma_\text{max} = \|\gv v_u\|_\infty^{-1}$. Recall that $u \in [1,N]$ is the (to-be-determined) index of the mode left uncorrected.

The natural figure of merit for a sensor with uncorrected noise is not the Fisher information: decoherence eventually causes $F(t)$ to peak and then decrease, rather than grow unbounded as in \eqref{eq:OptimalQFI}. Instead, it is sensitivity, defined as the smallest resolvable signal (i.e., giving unit signal-to-noise) per unit time \cite{degen2017quantum}.  For a single qubit with an energy gap $A \omega$ and dephasing time $T_2/B$, the best (smallest) achievable sensitivity is \cite{layden2018spatial}
\begin{equation}
    \eta = \min_{t > 0} \frac{1}{\sqrt{F(t)/t}}
    =
    \frac{\sqrt{B}}{A} \sqrt{\frac{2 e}{T_2}}.
    \label{eq:eta}
\end{equation}
Taking $\mathfrak h_j=1$ in \eqref{eq:dephasing_lindblad}, each physical qubit ($A=B=1$) gives $\eta_1 = \sqrt{2 e/T_2}$. $N$ such qubits operated in parallel give $\eta_\text{par} = \eta_1/\sqrt{N}$, while for entangled states one could reach $A=N$, often at the cost of an increased $B$. For example, one can easily show that a Greenberger-Horne-Zeilinger (GHZ) sensing scheme with the same $N$ qubits gives 
\begin{equation}
    \eta_\text{GHZ} = 
    \frac{\|D_C^{1/2} V^\top \frakh \|_2}{N} \sqrt{\frac{2 e}{T_2}},
    \label{eq:eta_GHZ}
\end{equation}
where $V = (\gv v_1, \dots, \gv v_N)$ and $D_C = \text{diag}(\lambda_1, \dots, \lambda_N)$ so that $C = V D_C V^\top$ \cite{GHZ}. 
Note that for uncorrelated noise we have $\|D_C^{1/2} V^\top \frakh\|_2=\sqrt{N}$, thus negating any gains from entanglement.

To find the sensitivity offered by the QEC code described above, we compute the sensor's effective Liouvillian under frequent recoveries: $\mathcal{L}_\text{eff} = \mathcal{R} \circ \mathcal{L} \circ \mathcal{P}$, where $\mathcal{R}$ is the QEC recovery, $\mathcal{L}$ is the sensor's Liouvillian [so that \eqref{eq:lindblad} reads $\dot{\rho} = \mathcal{L}(\rho)$], and $\mathcal{P}(\rho) = P \rho P$ \cite{layden2018spatial}. The usual QEC recovery procedure (i.e., the transpose channel) results in population leakage out of the codespace due to the uncorrected error $L_u$ when applied to the above code, even when $\Delta t \rightarrow 0$, which complicates the analysis \cite{nielsen2010quantum, lidar:2013}. To prevent such leakage at leading order in $\Delta t/T_2$, we modify the usual recovery so that the state is returned to the codespace after an error $L_u$, though perhaps with a logical error. This modification results in a Markovian, trace-preserving effective dynamics over the two-dimensional codespace, given by $\mathcal{L}_\text{eff}$. Specifically, the effective dynamics of the sensor becomes that of a dephasing qubit with $A = \gamma \gv |\gv v_u \cdot \frakh |$ and $B = \gamma^2 \lambda_u$, giving the closed-form expression $\eta_\text{QEC}^{(u)} = \eta_1 \sqrt{\lambda_u} /|\gv v_u \cdot \frakh | $. The optimal choice of $u$ is the one that minimizes this quantity, giving:
\begin{equation}
    \eta_\text{QEC}
    =\frac1{
    \|D_C^{-1/2} V^\top \frakh \|_\infty} 
    \sqrt{
    \frac{2e}{T_2}
    },
    \label{eq:eta_QEC}
\end{equation}
valid for arbitrary noise correlation profile $C$ \cite{regularize}. The calculation is straightforward but lengthy, and is given in \cite{SM} (\appref{app:dephasing_Leff}). Notice that the free parameter $\gamma$ cancels out in the final expression for $\eta_\text{QEC}$.

\eqref{eq:eta_QEC} allows one to determine the $C$'s for which this QEC scheme provides enhanced sensitivity over parallel and GHZ sensing. Notice that while HNLS is satisfied only in a measure-zero set of $C$'s (on the boundary of the set of possible correlation matrices), QEC can enhance sensitivity over a much larger set, regardless of whether it can approach the Heisenberg limit in $t$.

Notice that \eqref{eq:eta_QEC} admits a broad range of $\eta_\text{QEC}$ vs.\ $N$ scalings. This breadth of possibilities is due to the critical dependence of $\eta_\text{QEC}$ on $C=C(N)$, which could grow with $N$ in myriad different ways. The same is true of the Fisher information in the HNLS limit. Consider, for instance, a sensor comprising $k$ clusters of $n=N/k$ qubits, where each cluster satisfies the HNLS condition, but where the noise has no inter-cluster correlations. In this case, one could use the code of  \eqref{eq:phase_code} to make a noiseless sensor from each cluster, and perform GHZ sensing at the logical level to get $\eta \propto 1/k \propto 1/N$ Heisenberg scaling. On the other hand, given an $N$-qubit sensor already satisfying HNLS, adding an additional qubit which shares no noise correlations with the others has no impact on $\eta_\text{QEC}$. An intermediate example between these extreme scalings is analyzed in Ref.~\cite{SM} (\appref{app:dephasing_ex}).

\paragraph*{Discussion.} We have shown that noiseless ancillas, while frequently invoked, are not required for a large family of error-corrected quantum sensing scenarios where the Hamiltonian and the noise operators all commute. Our proof is constructive, and gives a numerical method for designing QEC codes for sensing through semidefinite programming, analogous to the techniques from Refs.\ \cite{fletcher:2007, kosut:2009} for quantum computing. Commuting noise, however, is not necessary for ancilla-free codes (see, e.g., Refs.\ \cite{dur2014improved, ozeri2013heisenberg, reiter:2017}); refining \thmref{thm:commuting} into a necessary and sufficient condition is a problem for future works.

We have also introduced near-optimal, closed-form QEC codes and associated recoveries for two common sensing scenarios. For dephasing qubits, we found an expression for the sensitivity enhancement offered by our QEC scheme under arbitrary Markovian noise, even when the Heisenberg limit in $t$ could not be reached. Our results raise the questions of whether there exists a simple geometric condition defining the set of $C$'s for which QEC can enhance sensitivity, and whether or not \eqref{eq:eta_QEC} is a fundamental bound for QEC schemes. Finally, ancilla-free QEC code design through convex optimization---both beyond the HNLS limit and for non-commuting noise more generally---presents a promising prospect for future work.

 \paragraph*{Acknowledgments.}
We thank Victor Albert, Kyungjoo Noh, Florentin Reiter and John Preskill for inspiring discussions. We acknowledge support from the ARL-CDQI (W911NF-15-2-0067, W911NF-18-2-0237), ARO (W911NF-18-1-0020, W911NF-18-1-0212), ARO MURI (W911NF-16-1-0349, W911NF-15-1-0548), AFOSR MURI (FA9550-14-1-0052, FA9550-15-1-0015), DOE (DE-SC0019406), NSF (EFMA-1640959, EFRI-ACQUIRE 1641064, EECS1702716) and the Packard Foundation (2013-39273).

\paragraph*{Author contributions.} S.Z.\ conceived of \thmref{thm:commuting}, the SDP to optimize the dephasing code, and the bosonic code. D.L.\ devised the dephasing code ansatz, its closed-form approximate solution, and its extension beyond HNLS. P.C.\ and L.J.\ supervised the project. All authors discussed the results and contributed to the final manuscript.

\bibliographystyle{apsrev}
\bibliography{ancilla-refs}

\begin{thebibliography}{51}
\bibitem[{reg()}]{regularize}
\bibinfo{note}{$D_C^{-1/2}$ is undefined when $C$ is singular. In this case,
  \eqref{eq:eta_QEC} should be regularized by replacing $D_C \rightarrow D_C +
  \epsilon I$, evaluating the norm, then taking $\epsilon \rightarrow 0$.}
\end{thebibliography}

\onecolumngrid
\vspace{0.3in}
\newpage
\appendix

\section{\label{app:necessary}An example where noiseless ancilla is necessary}

It is known that when the system dimension  ($d$) is sufficiently large, compared to the dimension of the noise space ($\dim \mathfrak S$), a valid QEC code satisfying \eqref{eq:QEC-2} always exists (Theorem 4 in \cite{knill2000theory}). 

However, it does not guarantee the existence of a valid QEC code for sensing, where \eqref{eq:QEC-1} and \eqref{eq:QEC-2} should be satisfied simultaneously. 
Here we provide an example where a valid QEC sensing code cannot be constructed without noiseless ancilla. Consider Gell-Mann matrices:
\begin{gather}
\lambda_1 = \begin{pmatrix} 0 & 1 & 0 \\ 1 & 0 & 0 \\ 0 & 0 & 0\end{pmatrix},\quad
\lambda_2 = \begin{pmatrix} 0 & -i & 0 \\ i & 0 & 0 \\ 0 & 0 & 0\end{pmatrix},\quad
\lambda_3 = \begin{pmatrix} 1 & 0 & 0 \\ 0 & -1 & 0 \\ 0 & 0 & 0\end{pmatrix},\\
\lambda_4 = \begin{pmatrix} 0 & 0 & 1 \\ 0 & 0 & 0 \\ 1 & 0 & 0\end{pmatrix},\quad
\lambda_5 = \begin{pmatrix} 0 & 0 & -i \\ 0 & 0 & 0 \\ i & 0 & 0\end{pmatrix},\quad\\
\lambda_6 = \begin{pmatrix} 0 & 0 & 0 \\ 0 & 0 & 1 \\ 0 & 1 & 0\end{pmatrix},\quad
\lambda_7 = \begin{pmatrix} 0 & 0 & 0 \\ 0 & 0 & -i \\ 0 & i & 0\end{pmatrix},\quad
\lambda_8 = \frac{1}{\sqrt{3}}\begin{pmatrix} 1 & 0 & 0 \\ 0 & 1 & 0 \\ 0 & 0 & -2\end{pmatrix}.
\end{gather}
and $\lambda_0$ is the identity matrix. 
The Hilbert space $\mathcal H = \mathcal H_3 \oplus \mathcal H_{d-3}$ is the direct sum of a $3$-dimensional and a $(d-3)$-dimensional Hilbert space. Let $H = \lambda_5 \oplus 0_{d-3}$ and $L_i = \lambda_i \oplus 0_{d-3}$ with $i=1,2,4$ where $0_i$ means a $i$-dimensional zero matrix. One can check that 
\begin{equation}
\mathfrak S = {\rm span}\{I,\lambda_i \oplus 0_{d-3},\;i=0,1,2,3,4,6,7,8\}
\end{equation}
then the HNLS condition $H \notin \mathfrak S$ is satisfied. Suppose we have a two-dimensional QEC sensing code 
\begin{equation}
\ket{0_\textsc{l}} = \alpha^0_3 \ket{0_3} + \alpha^0_{d-3} \ket{0_{d-3}},\quad \ket{1_\textsc{l}} = \alpha^1_3 \ket{1_3} + \alpha^1_{d-3} \ket{1_{d-3}}.
\end{equation}
where $\ket{0(1)_3} \in \mathcal H_3$ and $\ket{0(1)_{d-3}} \in \mathcal H_{d-3}$. First of all, we note that $\alpha_3^{0}$ and $\alpha_3^{1}$ are not both zero because $PHP \not\propto P$. If $\alpha_3^0 = 0$, due to the error correction condition $P L_i P \propto P$ and $P L_i^\dagger L_j P \propto P$, we must have 
\begin{equation}
\braket{1_3|\lambda_i|1_3} = 0,\quad i=1,2,3,4,6,7,8,
\end{equation}
leading to $\ket{1_3} = 0$. Therefore, we conclude that $\alpha_3^0$ and $\alpha_3^1$ are both non-zero. In this case we must have
\begin{equation}
\braket{0_3|\lambda_i|1_3} = 0,\quad i=1,2,3,4,6,7,8,
\end{equation}
which again could not be satisfied for non-zero $\ket{0(1)_3}$. Therefore we conclude that a valid QEC sensing code does not exist without noiseless ancilla. The interesting part about this example is that the dimension $d$ of the Hilbert space $\mathcal H$ could be arbitrary large compared to the number of noise operators $\dim \mathfrak S = 9$, yet there is no valid QEC code correcting noise and preserving signal simultaneously. It means that the role of noiseless ancilla in quantum sensing could not be replaced by a simple extension of the system dimension, as in traditional quantum error correction.

\section{\label{app:rotating} Commuting noise allows QEC to be performed with a lower frequency}

The commuting noise allows the QEC sensing protocol to be performed with a lower frequency, and therefore rendering it more accessible in experiments. Roughly speaking, in order to enhance the parameter estimation using QEC, it needs to be performed fast enough such that the time interval $\Delta t$ between every QEC is much smaller than $\omega^{-1}$ and $\kappa^{-1}$ ($\kappa$ is the noise strength) \cite{kessler2014quantum,zhou2018achieving}. A common technique to decrease the value of $\omega$ is to apply a opposite Hamiltonian $-\tilde{\omega} H$, where $\tilde{\omega} \approx \omega$ is close to $\omega$ \cite{sekatski2017quantum,demkowicz2017adaptive,yuan2016sequential,pang2017optimal}. For commuting noise, however, we could simply perform the QEC in an interaction picture where the signal Hamiltonian is sufficiently small such that we only need $\Delta t \ll \kappa^{-1}$. 
Consider the following evolution of $\rho$, 
\begin{equation}
\frac{d\rho}{dt} = -i[H_0+H_1,\rho] + \sum_{i} \kappa_i \left(L_i \rho L_i^\dagger - \frac{1}{2}\{L_i^\dagger L_i,\rho\}\right),
\end{equation}
where $H_0$ and $H_1$ are independent of time, $\norm{H_0} \gg \norm{H_1}$ and $H_0$ does not depend on the parameter $\omega$ we want to estimate. 
Usually, we take $H_0 = \tilde \omega H$ and $H_1 = \omega H$ where $\tilde \omega$ is close to $\omega$. 
To achieve an enhancement in parameter estimation by QEC, we require the time interval $\Delta t$ between each QEC recovery to satisfy 
\begin{equation}
\label{eq:cond}
\Delta t \cdot \Big\|\sum_{i} \kappa_i L_i^\dagger L_i\Big\| \ll 1\quad\text{~~and~~}\quad
\Delta t \cdot \norm{H_0+H_1} \ll 1,
\end{equation}
where the latter condition could be more difficult to satisfy. Here we show that for commuting noise, we can perform QEC in the interaction picture, relaxing the requirement on $\Delta t$ (\eqref{eq:cond}) to 
\begin{equation}
\label{eq:cond-rot}
\Delta t \cdot \Big\|\sum_{i} \kappa_i L_i^\dagger L_i\Big\| \ll 1\quad\text{~~and~~}\quad
\Delta t \cdot \norm{H_1} \ll 1,
\end{equation}
which is more accessible in experiments. 

Assuming $\tilde \rho(t) = e^{iH_0 t} \rho e^{-iH_0 t}$, then the evolution in the interaction picture is
\begin{equation}
\frac{d\tilde{\rho}}{dt} = -i[H_1(t),\tilde{\rho}] + \sum_{i} \kappa_i \left(L_i(t) \tilde\rho L_i^\dagger(t) - \frac{1}{2}\{L_i^\dagger(t) L_i(t),\tilde\rho\}\right),
\end{equation}
where $H_1(t) = e^{i H_0 t}H_1 e^{-i H_0 t}$ and $L_i(t) = e^{i H_0 t}H e^{-i H_0 t}$. 

Consider the evolution of $\tilde{\rho}$ in time interval $\Delta t$ satisfying \eqref{eq:cond-rot} such that $\tilde{\rho}(t+\Delta t) \approx \tilde{\rho}(t)$, but $\norm{H_0} \Delta t$ is not negligible so that $\rho(t+\Delta t) \not\approx \rho(t)$. Then
\begin{equation}
\tilde{\rho}(t + \Delta t) = \tilde{\rho}(t) + \int_{t}^{t+\Delta t} dt \bigg( -i [H_1(t'),\tilde\rho(t)] + \sum_i \kappa_i \bigg(L_i(t') \tilde\rho(t) L_i^\dagger(t') - \frac{1}{2}\{L_i^\dagger(t') L_i(t'),\tilde\rho(t)\}\bigg) \bigg) + O(\Delta t^2).
\end{equation}
If we can find a QEC code whose projector $P$ satisfies 
\begin{equation}
\label{eq:QEC-2-rotating}
P L_i(t') P = \lambda_i(t') P   \quad\text{~and~}\quad P L_i(t')^\dagger L_j(t'') P = \mu_{ij}(t-t'') P + O(\Delta t^2),
\end{equation}
where $\lambda_i(t')$ and $\mu_{ij}(t-t'')$ are constants for arbitrary $t'$ and $t''$ in $[t,t+\Delta t]$, one can do quantum error corretion on $\tilde{\rho}(t + \Delta t)$ in the interaction picture, resulting in the effective dynamics
\begin{equation}
\mathcal{R}(\tilde{\rho}(t + \Delta t)) = \tilde\rho(t) -i\Big[\int_t^{t+\Delta t}dt' P H_1(t') P,\tilde{\rho}(t)\Big] + O(\Delta t^2).
\end{equation}
where $\mathcal R$ is the recovery channel. We also want the effective Hamiltonian 
\begin{equation}
\label{eq:QEC-1-rotating}
\int_t^{t+\Delta t}dt' P H_1(t') P \not\propto P 
\end{equation}
to be nontrivial.

In general, it could be difficult to find a QEC code satisfying \eqref{eq:QEC-2-rotating} and \eqref{eq:QEC-1-rotating}. However, when we have a commuting noise model, i.e., $[H_0,L_i] = [H_1,L_i] = [H_0,H_1] = [L_i,L_j] = 0$ for all $i$ and $j$, \eqref{eq:QEC-2-rotating} and \eqref{eq:QEC-1-rotating} are simply reduced to the HNLS condition (\eqref{eq:QEC-1} and \eqref{eq:QEC-2}), because $L_i(t) = L_i$ and $H_1(t) = H_1$. Therefore we can perform the QEC in the interaction picture by using the same QEC code as in the Schr\"odinger picture and rotating it accordingly.

\section{\label{app:dephasing_KL}Validity and optimality of dephasing code}

We first verify that the dephasing code 
\begin{equation}
\ket{0_\textsc{l}} = \bigotimes_{j=1}^N
\Big[
\cos(\theta_j) \ket{0_j}
+
i \sin(\theta_j) \ket{1_j}
\Big],\qquad
\ket{1_\textsc{l}} = X^{\otimes N} \ket{0_\textsc{l}}
\end{equation}
corrects all the noise in the Lindblad span [\eqref{eq:QEC-2}] when HNLS is satisfied. We begin by showing that $\ket{0_\textsc{l}}$ and $\ket{1_\textsc{l}}$ are orthonormal. Normalization is clear. Orthogonality is apparent by noting that the components of $\ket{0_\textsc{l}}$ and $\ket{1_\textsc{l}}$ on qubit $j$ are
\begin{equation}
\ket{0_{\textsc{l},j}}
=
\cos(\theta_j) \ket{0_j} + i \sin(\theta_j) \ket{1_j}
,\qquad
\ket{1_{\textsc{l},j}}
=
i \sin(\theta_j) \ket{0_j} + \cos(\theta_j) \ket{1_j}
\end{equation}
respectively, so that $\ket{0_\textsc{l}} = \otimes_{j=1}^N \ket{0_{\textsc{l},j}}$ and $\ket{1_\textsc{l}} = \otimes_{j=1}^N \ket{1_{\textsc{l},j}}$. Clearly $\ket{0_{\textsc{l},j}}$ and $\ket{1_{\textsc{l},j}}$ are orthogonal for all $j$, so $\ket{0_\textsc{l}}$ and $\ket{1_\textsc{l}}$ are also orthogonal. Note that one could choose a different phase in $\ket{0_{\textsc{l},j}}$, provided $\ket{1_{\textsc{l},j}}$ is redefined appropriately.

Next, we examine terms of the form $P (\gv v \cdot \gv Z) P$. For $N \ge 2$, the orthogonality of $\ket{0_{\textsc{l},j}}$ and $\ket{1_{\textsc{l},j}}$ implies that $\bra{0_\textsc{l}} Z_i \ket{1_\textsc{l}} =0$. On the other hand, 
\begin{equation}
    \bra{0_\textsc{l}} Z_i \ket{0_\textsc{l}} = \cos(2 \theta_i).
\end{equation}
Likewise, 
\begin{equation}
    \bra{1_\textsc{l}} Z_i \ket{1_\textsc{l}} =
    \bra{0_\textsc{l}} X^{\otimes N} Z_i X^{\otimes N} \ket{0_\textsc{l}}=
    -\bra{0_\textsc{l}}  Z_i  \ket{0_\textsc{l}} =
    -\cos(2 \theta_i).
\end{equation}
Therefore
\begin{equation}
    P(\gv v \cdot \gv Z) P = \gv v \cdot \cos(2 \gv \theta) \, Z_\textsc{l}
    =
    (\gv v \cdot \frakb^\diamond )  Z_\textsc{l}, 
    \label{eq:KL1}
\end{equation}
where the cosine is taken element-wise, $Z_\textsc{l} = \ket{0_\textsc{l}}\!\!\bra{0_\textsc{l}} - \ket{1_\textsc{l}}\!\!\bra{1_\textsc{l}}$, and we have used $\gv \theta = \frac{1}{2}\arccos \frakb^\diamond$. If $\gv v$ is an eigenvector of $C$ with non-zero eigenvalue, then $\gv v \in {\rm col}(C)$ and $\frakb^\diamond \perp {\rm col}(C)$ gives $P (\gv v \cdot \gv Z) P = 0$. 

We now consider terms of the form $P (\gv v \cdot \gv Z) (\gv u \cdot \gv Z) P$. For $N \ge 3$, the orthogonality of $\ket{0_{\textsc{l},i}}$ and $\ket{1_{\textsc{l},i}}$ implies that $\bra{0_\textsc{l}} Z_j Z_k  \ket{1_\textsc{l}} =0$. We also have
\begin{equation}
\bra{1_\textsc{l}} Z_j Z_k \ket{1_\textsc{l}}
=
\bra{0_\textsc{l}} (X^{\otimes N} Z_j X^{\otimes N} ) ( X^{\otimes N} Z_k X^{\otimes N} ) \ket{0_\textsc{l}}
=
(-1)^2 \bra{0_\textsc{l}} Z_j Z_k \ket{0_\textsc{l}}.
\end{equation}
Therefore, $P (\gv v \cdot \gv Z) (\gv u \cdot \gv Z) P \propto P$ for any $\gv v$ and $\gv u$. This shows that \eqref{eq:QEC-2} is satisfied by our dephasing code. 

Next we prove that the dephasing code given by the SDP in Eqs.~(\ref{eq:qubit-SDP-1})--(\ref{eq:qubit-SDP-2}) saturates \eqref{eq:OptimalQFI}. This in turn guarantees that it satisfies \eqref{eq:QEC-1}. 
According to the proof of \thmref{thm:commuting}, the optimal asymptotic QFI is 
\begin{equation}
F_{\rm opt}(t) = t^2 \min_{\gv \ell \in \mathfrak S_{\rm diag}} \Big\|\sum_{i=1}^N \mathfrak h_i \gv{z}_i + \gv \ell\Big\|^2_\infty \leq t^2 \min_{\gv \ell \in \mathfrak S_{\textsc v}} \Big\|\sum_{i=1}^N \mathfrak h_i \gv{z}_i + \gv \ell\Big\|^2_\infty \equiv F'_{\rm opt}(t),
\end{equation}
where $\mathfrak S_{\textsc v} = {\rm span}\{\sum_{j=1}^N (v_i)_j \gv{z_j},i\} \subseteq \mathfrak S_{\rm diag}$ and $\gv{z_j} \in \mathbb{R}^{2^N}$ is defined such that $Z_j = \text{diag}(\gv z_j)$.
A straightforward manipulation of the previous equation gives
\begin{equation}
F'_{\rm opt}(t) = t^2 \min_{\gv v \in {\rm col}(C)} \| \frakh + \gv v \|^2_{1}.
\label{eq:app_phase_F}
\end{equation}
Let $\frakb^\diamond$ be the optimal solution of the SDP given by Eqs.~(\ref{eq:qubit-SDP-1})--(\ref{eq:qubit-SDP-2}), then the dephasing code satisfies
\begin{equation}
H_\text{eff} = \frac{1}{2} P (\frakh \cdot \gv Z) P = \frac{1}{2} (\frakh \cdot \frakb^\diamond) Z_\textsc{l}.
\end{equation}
Now we prove the optimal value of Eqs.~(\ref{eq:qubit-SDP-1})--(\ref{eq:qubit-SDP-2}), $\frakh \cdot \frakb^\diamond$, is equal to 
\begin{equation}
\min_{\gv v \in {\rm col}(C)} \| \frakh + \gv v \|_{1}
\end{equation}
from \eqref{eq:app_phase_F}. It implies $F_{\rm opt}(t) = F'_{\rm opt}(t)$ is attainable using the dephasing code. This can be easily shown by noting that the dual of Eqs.~(\ref{eq:qubit-SDP-1})--(\ref{eq:qubit-SDP-2}) is \cite{boyd2004convex}:
\begin{equation}
\begin{split}
g(\lambda,\gv{v}) &= \max_{\frakh' \in \mathbb{R}^N} L(\frakh',\lambda,\gv{v}) = \max_{\frakh' \in \mathbb{R}^N} \frakh' \cdot(\frakh + \gv{v}) - \lambda \|\frakh'\|_\infty + \lambda \\ 
&= \begin{cases}
\lambda, & \lambda \geq \|\frakh + \gv{v}\|_1 \\
\infty, & \lambda < \|\frakh + \gv{v}\|_1
\end{cases},
\end{split}
\end{equation}
and the minimum of the dual is  
\begin{equation}
\min_{\lambda,\gv{v} \in {\rm col}(C)} g(\lambda,\gv{v})= \min_{\gv v \in {\rm col}(C)} \| \frakh + \gv v \|_{1},
\end{equation}
proving the optimality of the dephasing code.

Similarly, for the approximately-optimal code, which uses $\gv \theta = \frac{1}{2} \arccos \tilde{\gv \frakb} ^\diamond$ for $\tilde{\gv \frakb} ^\diamond$ from \eqref{eq:hC_tilde}, we have
\begin{equation}
P (\gv v \cdot \gv Z) P
=
(\gv v \cdot \tilde{\gv \frakb} ^\diamond) Z_\textsc{l} = 0
\end{equation}
for any $\gv v \in \text{col}(C)$, since $\tilde{\gv \frakb} ^\diamond \in \text{ker}(C)$. Therefore \eqref{eq:QEC-2} is exactly satisfied when $H \notin \mathfrak{G}$. Moreover, in this case 
\begin{equation}
H_\text{eff} = P H P = \frac{\gamma}{2} (\frakh \cdot \text{proj}_{\text{ker}(C)}\frakh ) Z_\textsc{l} \not \propto P
\end{equation}
as claimed. Therefore, the only disadvantage of using $\tilde{\frakb}^\diamond$ rather than $\frakb^\diamond$ for the code is that the former may give a slower accumulation of signal at the logical level. Note that the range of allowed $\gamma$'s is always finite when $H \notin \mathfrak{G}$, and is set by the requirements that the arccosine function be defined for each element, and that $H_\text{eff} \not \propto P$. Finally, the above calculation is nearly identical if we instead use $\gv \theta = \frac{1}{2} \arccos(\gamma \gv v_u)$, as in \eqref{eq:theta_nonHNLS}, for when the HNLS condition is not satisfied. The main effect of replacing $\text{proj}_{\text{ker}(C)} \gv \frakh \rightarrow \gv v_u$ to go beyond HNLS is that $P L_u P \propto Z_\textsc{l}$ with a non-vanishing coefficient. This makes $L_u$ an uncorrectable error, as expected.

\section{\label{app:bosonic-noise} Effective Lindblad span in lossy bosonic channel}

Consider the lossy bosonic channel
\begin{equation}
\frac{d\rho}{dt} = -i[\sum_{i=1}^s \zeta_i (a^\dagger a)^i, \rho] + \kappa \left(a\rho a^\dagger - \frac{1}{2}\{a^\dagger a,\rho\}\right).
\end{equation}
Using an arbitrary quantum error correction code to correct the boson loss, the effective dynamics up to the lowest order would be 
\begin{equation}
\frac{d\rho}{dt} = -i[\sum_{i=1}^s P \zeta_i (a^\dagger a)^i P,\rho].
\end{equation}
To sense $\omega$, we would like to filter out $(a^\dagger a)^i$ up to $s-1$ such that they act trivial in the code space. Therefore we also require
\begin{equation}
P (a^\dagger a)^i P = \lambda_i P 
\end{equation}
for $1 \leq i \leq s-1$ and arbitrary $\lambda_i$. This is equivalent to use the following Lindblad span
\begin{equation}
\mathfrak S = \{I,a,a^\dagger,(a^\dagger a)^i,1 \leq i \leq s-1\}.
\end{equation}
and the signal Hamiltonian is $(a^\dagger a)^s$. 

Note that this is not a commuting noise model because $[a,a^\dagger] \neq 0$. Therefore, we reconsider it in the setting of \appref{app:rotating} where QEC is performed in the interaction picture. First, we assume the code corrects all noise in the Lindblad span $\mathfrak S$. We also assume
\begin{equation}
P\ket{m-1}\bra{m}P = 0, 
\end{equation}
which is easy to satisfy and true for the Chebyshev code by our construction (see \appref{app:bosonic-code}). 

Let $H_0 = \sum_{i=1}^{s-1} \zeta_i (a^\dagger a)^i$ and $H_1 = \zeta_s (a^\dagger a)^s$. Then
\begin{equation}
P a(t) P = P e^{i H_0 t} a e^{-i H_0 t} P = P \left(\sum_{m=1}^\infty e^{-it\sum_{i=1}^{s-1}\zeta_i(m^i - (m-1)^i)}\sqrt{m}\ket{m-1}\bra{m}\right)P = 0,
\end{equation}
and 
\begin{equation}
\begin{split}
P a^\dagger(t) a(t') P 
= P \left(\sum_{m=1}^\infty e^{i (t-t')\sum_{i=1}^{s-1} \zeta_i \left(m^i - (m-1)^i\right)} m \ket{m}\bra{m} \right)P = \lambda (t-t') P + O\left(\Big(\Delta t \cdot \frac{\norm{H_0}}{M}\Big)^2\right).
\end{split}
\end{equation}
for some constant $\lambda$. According to \eqref{eq:QEC-2-rotating}, an enhancement by QEC could be achieved when $\Delta t$
\begin{equation}
\Delta t \cdot \norm{\kappa a^\dagger a} \ll 1 \text{~~~and~~~} \Delta t \cdot \frac{1}{M}\norm{H_0} \ll 1.
\end{equation}
It means that by performing the QEC in the interaction picture, we could relax the second condition in \eqref{eq:cond} by a factor of $M$, the upper bound of the number of bosons in the channel.

\section{\label{app:bosonic-code} Validity and near-optimality of the Chebyshev code}

We first provide two equalities which will be used later in showing the validity and optimality of the Chebyshev code: 
\begin{lemma}
\label{thm:chebyshev}
Suppose $s$ is an integer larger than one. Then we have
\begin{equation}
\sum_{k=0}^s (-1)^k \abs{c_k}^2 \left(\sin\frac{k\pi}{2s}\right)^{2i} = 0,\quad\forall 1 \leq i \leq s-1,
\end{equation}
and 
\begin{equation}
\sum_{k=0}^s (-1)^k \abs{c_k}^2 \left(\sin\frac{k\pi}{2s}\right)^{2s} = \frac{(-1)^s}{2^{2s-2}}, 
\end{equation}
where $\abs{c_k}^2 = \frac{2}{s} - \frac{1}{s}\delta_{k0} - \frac{1}{s}\delta_{ks}$.
\end{lemma}

\begin{proof}
We first notice that for all $0 \leq \ell \leq s-1$,
\begin{equation}
\sum_{k=0}^{s-1} (-1)^k \cos\frac{k\ell\pi}{s} = {\rm Re}\left[\sum_{k=0}^{s-1} e^{i\big(\frac{k\ell\pi}{s}+k\pi\big)}\right] = {\rm Re}\left[ \frac{1 - (-1)^{s+\ell}}{1 + e^{i\frac{\ell\pi}{s}}}\right] = \frac{1+(-1)^{s+\ell+1}}{2},
\end{equation}
which only depends on the parity of $\ell$. Then we have
\begin{equation}
\sum_{k=0}^{s-1} (-1)^k \left(\cos\frac{k\pi}{s}\right)^\ell = \frac{1}{2^\ell} \sum_{k=0}^{s-1} (-1)^k \sum_{j=0}^\ell \binom{\ell}{j} \cos\frac{(2j-\ell)k\pi}{s} = \frac{1+(-1)^{s+\ell+1}}{2}.
\end{equation}
When $\ell = s$, 
\begin{equation}
\sum_{k=0}^{s-1} (-1)^k \left(\cos \frac{k\pi}{s}\right)^s = \frac{1}{2^\ell} \sum_{k=0}^{s-1} (-1)^k \sum_{j=0}^s \binom{s}{j} \cos\frac{(2j-s)k\pi}{s} = \frac{s}{2^{2s-1}}.
\end{equation}
Therefore,
\begin{equation}
\begin{split}
\frac{1}{2^i}\sum_{k=0}^{s-1} (-1)^k\left(1-\cos\frac{k\pi}{s}\right)^{i} &= \frac{1}{2^i}\sum_{k=0}^{s-1} (-1)^k \sum_{\ell=0}^i \binom{i}{\ell} (-1)^\ell \left(\cos\frac{k\pi}{s}\right)^\ell\\
&= \frac{1}{2^i}\sum_{\ell=0}^i \binom{i}{\ell} (-1)^\ell \left(\frac{1+(-1)^{s+\ell+1}}{2} \right) = \frac{(-1)^{s+1}}{2},
\end{split}
\end{equation}
and 
\begin{equation}
\begin{split}
\frac{1}{2^s}\sum_{k=0}^{s-1} (-1)^k\left(1-\cos\frac{k\pi}{s}\right)^{s} &= \frac{1}{2^s}\sum_{k=0}^{s-1} (-1)^k \sum_{\ell=0}^s \binom{s}{\ell} (-1)^\ell \left(\cos\frac{k\pi}{s}\right)^\ell\\
&= \frac{(-1)^{s+1}}{2} + \frac{s(-1)^s}{2^{2s-1}}.
\end{split}
\end{equation}
As a result, when $1 \leq i \leq s-1$
\begin{equation}
\sum_{k=0}^s (-1)^k \abs{c_k}^2 \left(\sin\frac{k\pi}{2s}\right)^{2i} = \frac{2}{s} \left(\sum_{k=0}^{s-1} (-1)^k\left(\sin\frac{k\pi}{2s}\right)^{2i} + \frac{(-1)^s}{2}\right)  = 0;
\end{equation}
and when $i=s$,
\begin{equation}
\sum_{k=0}^s (-1)^k \abs{c_k}^2 \left(\sin\frac{k\pi}{2s}\right)^{2s} = \frac{2}{s} \left(\sum_{k=0}^{s-1} (-1)^k\left(\sin\frac{k\pi}{2s}\right)^{2s} + \frac{(-1)^s}{2}\right)  = 4\left(\frac{-1}{4}\right)^s. 
\end{equation}
\end{proof}

The $s$-th order Chebyshev code
\begin{equation}
\ket{0_\textsc{l}} = \sum_{k \text{~even}} \tilde{c}_k \Ket{\left\lfloor M \sin^2\left( \frac{k\pi}{2s}\right) \right\rfloor},\quad 
\ket{1_\textsc{l}} = \sum_{k \text{~odd}} \tilde{c}_k \Ket{\left\lfloor M \sin^2\left( \frac{k\pi}{2s}\right) \right\rfloor}
\end{equation}
should be capable of correcting the Lindblad span
\begin{equation}
\mathfrak S = {\rm span}\{I,a,a^{\dagger},(a^\dagger a)^i,\forall 1\leq i \leq s-1\}.
\end{equation}

To correct the off-diagonal noise $a^i$ and $(a^{\dagger})^i$ up to $1 \leq i \leq t$, we simply need the distance between $\ket{0_\textsc{l}}$ and $\ket{1_\textsc{l}}$ defined by 
\begin{equation}
\textrm{dist}(\ket{0_\textsc{l}},\ket{1_\textsc{l}}) = \min_{\substack{m_0,m_1 \in \mathbb N,\\\braket{m_0|0_\textsc{l}}^2 > 0,\,\braket{m_1|1_\textsc{l}}^2 > 0}} \abs{m_0 - m_1}
\end{equation}
is larger than $2t+1$. Note that
\begin{equation}
\textrm{dist}(\ket{0_\textsc{l}},\ket{1_\textsc{l}}) \geq M \sin^2\left( \frac{\pi}{2s}\right) - 1,
\end{equation}
so the off-diagonal noise can be corrected as long as $M \sin^2 \frac{\pi}{s} \geq 2(t+1)$. Particularly, when $t=1$, we only need $\textrm{dist}(\ket{0_\textsc{l}},\ket{1_\textsc{l}}) \geq 3$, or $M \sin^2\left( \frac{\pi}{s}\right) \geq 4$.

To correct the diagonal noise $(a^\dagger a)^i$ for $1 \leq i \leq s-1$, we simply need to choose a suitbale $\{\abs{c_k}^2\}_{k=0}^{s}$ to satisfy the following $s-1$ equations
\begin{equation}
\braket{0_\textsc{l}|(a^\dagger a)^i|0_\textsc{l}} - \braket{1_\textsc{l}|(a^\dagger a)^i|1_\textsc{l}}  = \sum_{k=0}^s (-1)^k \tilde{c}_k^2 \left\lfloor M \sin^2\left(\frac{k\pi}{2s}\right) \right\rfloor^{2i} = 0,
\end{equation}
and
\begin{equation}
\sum_{k=0}^{s} (-1)^k\abs{\tilde{c}_k}^2 = 0,
\quad
\sum_{k=0}^{s} \abs{\tilde{c}_k}^2 = 2. 
\end{equation}
The linear system of equations could be written as  
\begin{equation}
\label{eq:linear}
\tilde{A} \;\gv{\tilde{c}} = \gv{e}
\end{equation}
where $\gv{\tilde{c}} = (\abs{\tilde{c}_0}^2\,\abs{\tilde{c}_1}^2\,\cdots\,\abs{\tilde{c}_{s}}^2)^T$, $\gv{e} = (0\,\cdots\,0\,1)^T$, $\tilde{A}$ is a $s+1$ by $s+1$ matrix $\tilde{A}_{ik} = (-1)^k\left\lfloor M \sin^2\frac{k\pi}{2s} \right\rfloor^{i}/M^{i}$ when $0 \leq i \leq s-1$ (we assume $0^0 = 1$) and $\tilde{A}_{sk} = 1$. \eqref{eq:linear} is solvable since $A$ is invertible, which proving the validity of the Chebyshev code.

\begin{figure}[tbp]
\includegraphics[width=11cm]{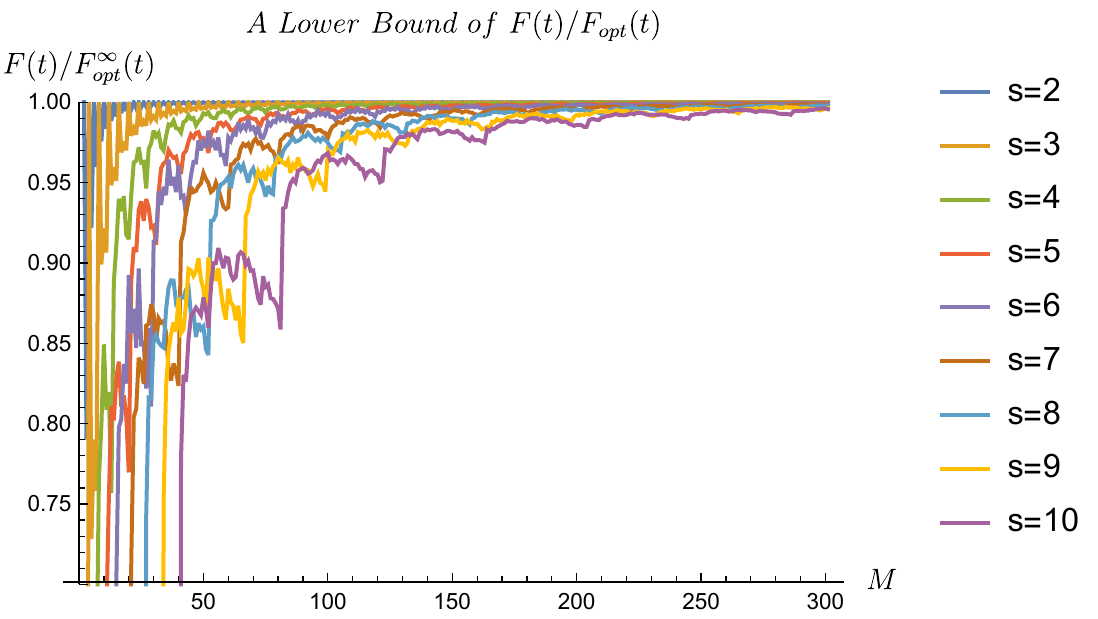}
\caption{\label{fig:nearopt} The near-optimality of the Chebyshev code. The horizonal axis indicates $M$, the largest photon allowed in the bosonic channel and the vertical axis indicates $F(t)/F^\infty_{\rm opt}(t)$, a lower bound of $F(t)/F_{\rm opt}(t)$. When $M$ is sufficiently large, $F(t)$ is very close to its optimal value $F_{\rm opt}(t)$. Also, when $s$ increases, we will need a larger $M$ to achieve the optimality.}
\end{figure}

Next we show the near-optimality of the Chebyshev code. First we calculate an upper bound of the optimal asymptotic QFI \eqref{eq:OptimalQFI}, since 
\begin{equation}
\begin{split}
\norm{(a^\dagger a)^s - \mathfrak S}
&= M^s \min_{\substack{\forall \chi_i \in \mathbb R}} \max_{\substack{k\in\mathbb Z,\\k\in[0,M]}} \bigg|\bigg(\frac{k}{M}\bigg)^s - \sum_{i=0}^{s-1} \chi_i \bigg(\frac{k}{M}\bigg)^i\bigg| \\
&\leq \bigg(\frac{M}{2}\bigg)^s \min_{\substack{\forall \chi_i \in \mathbb R}} \max_{x\in[-1,1]} \Big|x^s - \sum_{i=0}^{s-1} \chi_i x^i\Big| = 2 \bigg(\frac{M}{4}\bigg)^{s},
\end{split}
\end{equation}
we have
\begin{equation}
F_{\rm opt}(t) = 4 t^2 \min_{S \in \mathfrak S} \norm{(a^\dagger a)^s - S}^2_\infty \leq 16 t^2 \bigg(\frac{M}{4}\bigg)^{2s} \equiv F^\infty_{\rm opt}(t). 
\end{equation}

According to \lemmaref{thm:chebyshev}, 
\begin{equation}
\sum_{k=0}^s (-1)^k c_k^2 \left( \sin^2 \frac{k\pi}{2s}\right)^{i} = 0 \;\Longrightarrow A \;\gv{c} = \gv{e},
\end{equation}
where $\gv{c} = (\abs{c_0}^2\,\abs{c_1}^2\,\cdots\,\abs{c_{s}}^2)^T$, $A_{ik} = (-1)^k\left(\sin^2\frac{k\pi}{2s} \right)^{i}$ when $0 \leq i \leq s-1$ and $\tilde{A}_{sk} = 1$. Note that 
\begin{equation}
\sin^2 \frac{k\pi}{2s} - \frac{1}{M} \leq \frac{\left\lfloor M \sin^2\left(\frac{k\pi}{2s}\right)\right\rfloor}{M} \leq \sin^2\left( \frac{k\pi}{2s}\right).
\end{equation}
As $M$ becomes sufficiently large, we have
\begin{equation}
\gv{\tilde c} = {\tilde A}^{-1} \gv e = {A + (\tilde A - A)}^{-1} \gv e = (I + A^{-1}(\tilde A - A))^{-1}\gv c = \gv c + O\left(\frac{1}{M}\right).
\end{equation}
Then
\begin{multline}
\braket{0_\textsc{l}|(a^\dagger a)^s|0_\textsc{l}} - \braket{1_\textsc{l}|(a^\dagger a)^s|1_\textsc{l}} = \sum_{k=0}^s (-1)^k \tilde{c}_k^2 \left\lfloor M \sin^2\frac{k\pi}{2s} \right\rfloor^{s} \\
= \frac{(-M)^s}{2^{2s-2}} + \sum_{k=0}^s (-1)^k \left(\tilde{c}_k^2 \left\lfloor M \sin^2\frac{k\pi}{2s} \right\rfloor^{s} - c_k^2 \left( M \sin^2\frac{k\pi}{2s} \right)^{s} \right) \geq \frac{(-M)^s}{2^{2s-2}} - M^s \norm{\gv c - \gv{\tilde c}}_2^2
\end{multline}
where $\norm{\gv c - \gv{\tilde c}}_2^2 = \sum_{k=0}^{s} \abs{\tilde{c}_k - c_k}^2$ is the two-norm. 
Therefore 
\begin{equation}
\frac{F_{\rm opt}(t) - F(t)}{F_{\rm opt}(t)} \leq 1 - \frac{F(t)}{F^\infty_{\rm opt}(t)} = O\left(\frac{1}{M^2}\right),
\end{equation}
proving the near-optimality of the Chebyshev code. The numerical value of $F(t)/F^\infty_{\rm opt}(t)$ is plotted in \figref{fig:nearopt} as a lower bound of $F(t)/F_{\rm opt}(t)$. 

Consider the $(s-1,\frac{M}{s}-1)$ binomial code (suppose $M$ is a multiple of $s$)~\cite{michael2016new}
\begin{equation}
\ket{0^b_{\textsc l}/1^b_{\textsc l}} = 
\sum_{k\text{~even/odd}}^{[0,s]} \frac{1}{\sqrt{2^{s-1}}}\sqrt{\binom{s}{k}} \Ket{\frac{k}{s}M}.
\end{equation}
We have
\begin{equation}
\begin{split}
\bra{0^b_{\textsc l}}(a^\dagger a)^\ell\ket{0^b_{\textsc l}} - \bra{1^b_{\textsc l}}(a^\dagger a)^\ell\ket{1^b_{\textsc l}} &= \frac{M^\ell}{2^{s-1}s^\ell} \sum_{k=0}^s \binom{s}{k} k^\ell (-1)^k 
\\
&= \frac{M^\ell}{2^{s-1}s^\ell} \left(x\frac{d}{dx}\right)^\ell (1+x)^{s}\bigg|_{x=-1} = \begin{cases}  0, & \ell = 1,\ldots,s-1,\\
\frac{(-1)^s s! M^s}{2^{s-1}s^s}, & \ell = s. 
\end{cases}
\end{split}
\end{equation}
Clearly the $(s-1,\frac{M}{s}-1)$ binomial code also corrects the Lindblad span, but the stregth of the signal is exponentially smaller with respect to $s$:
\begin{equation}
\frac{F^b(t)}{F_{\rm opt}(t)} \approx \left(\frac{2^{s-1}s!}{s^s}\right)^2 = O\left(s\left(\frac{2}{e}\right)^{2s}\right).
\end{equation}

\section{\label{app:dephasing_Leff}Effective dynamics under the dephasing code}

As discussed in the main text, we find the achievable sensitivity offered by the dephasing qubit code outside the HNLS limit by computing its effective dynamics under frequent QEC recoveries \cite{layden2018spatial}. For $\Delta t \ll T_2$, this dynamics is generated by the effective Liouvillian $\mathcal{L}_\text{eff} = \mathcal{R} \circ \mathcal{L} \circ \mathcal{P}$, which we construct explicitly here for the qubit code defined by Eqs.\ (\ref{eq:phase_code}) and (\ref{eq:theta_nonHNLS}).

We begin by constructing the recovery operation $\mathcal{R}$. The standard QEC recovery procedure (i.e., the transpose channel) is the following: The projector $P$ onto the codespace, together with the correctable errors for our code $\{ L_i\}_{i \neq u}$, define a set of rank-two projectors $\{ P_i \}_{i \neq u}$ and corresponding unitaries $\{ U_i \}_{i \neq u}$ \cite{zhou2018achieving}. To leading order in $\Delta t/T_2$, the correctable errors cause the sensor to jump into the subspaces defined by the $P_i$'s. The standard/transpose recovery consists of measuring in $\{ P, P_1, P_2, \dots \}$ (not containing a $P_u$), and applying $U_i^\dagger$ if the state is found in col$(P_i)$, for $i \neq u$ \cite{nielsen2010quantum, lidar:2013}. In the present setting, however, this procedure is problematic. The issue is that the uncorrected error $L_u$ can cause the state to jump into the ``remainder" subspace, with projector $P_\textsc{r} = I - P - \sum_{i\neq u} P_i$. To avoid population leakage from the codespace into col($P_\textsc{r}$) at leading order in $\Delta t/T_2$, we modify the usual procedure by returning the state to the codespace in the event of an error $L_u$, even though this error cannot---by design---be fully corrected. This gives an $\mathcal{L}_\text{eff}$ with non-trivial dynamics only in the 2-dimensional codespace, which lets us analyze the sensitivity using \eqref{eq:eta}, i.e., as though it were a two-level system.

Concretely, let us first define Knill-Laflamme coefficients $m_{ij} \in \mathbb{R}$ by $P L_i^\dagger L_j P = m_{ij} P$, for all $i, j \in {1, \dots, N}$. We also define the $N\times N$ matrix $M = (m_{ij})_{i,j=1}^N$, and the $(N-1)\times(N-1)$ submatrix $\tilde{M} = (m_{ij})_{i,j \neq u}$, which is equal to $M$ with the $u^\text{th}$ row and column removed. [Recall that $u$ is the index of the noise mode left uncorrected, as per \eqref{eq:theta_nonHNLS}.] Then, let $W$ be an orthogonal matrix diagonalizing $\tilde{M}$, such that $W^\top \tilde{M} W = \text{diag}(d_1, d_2, \dots)$. This $W$ lets us define new error operators $\{F_i\}_{i\neq u}$ such that
\begin{equation}
F_k = \sum_{i\neq u} w_{ik} L_i
\qquad \text{and} \qquad
L_j = \sum_{k \neq u} w_{jk} F_k,
\end{equation}
which satisfy $P F_i F_j P = \delta_{ij} d_i P$. For $j\neq u$, we then use $F_j$ to define a unitary $U_j$ via polar decomposition, such that $F_j P = \sqrt{d_j} U_j P$, and finally $P_j = U_j P U_j^\dagger$. (If $d_j=0$ take $U_j = 0$.) So far we have followed Refs.\ \cite{nielsen2010quantum, lidar:2013}. However, we now define an additional unitary $U_u$ via the polar decomposition of $P_R L_uP$:
\begin{align}
P_R L_u P &= U_u \sqrt{(P_R L_u P)^\dagger (P_R L_u P)} = \sqrt{m_{uu} - \lambda_u \gamma^2 - \alpha} \; U_u P,
\label{eq:u_polar}
\end{align}
(taking $U_u=0$ if $P_R L_uP=0$) for some constant $\alpha$. Concretely, $\alpha$ is defined through
\begin{align}
    P L_u^\dagger P_\textsc{r} L_u P
    &=
    P L_u^2 P - (P L_u P)^2 - \sum_{i \in \mathcal{I}} P L_u U_i P U_i^\dagger L_u P\\
    &=
    (m_{uu} - \lambda_u \gamma^2) P 
    -
    \sum_{i \in \mathcal{I}} \frac{1}{|d_i|} P L_u L_i P L_i L_u P \nonumber \\
    &=:
    (m_{uu} - \lambda_u \gamma^2 - \alpha) P. \nonumber
\end{align}
where $\mathcal{I}=\{i \, |\,  i \neq u,\, d_i \neq 0 \text{ and } 1 \le i \le N \}$. Our modified recovery channel then consists of measuring in $\{P, P_1, \dots, P_N\}$, where $P_j \equiv U_j P U_j^\dagger$ for all $j \in \{1, \dots, N\}$ (now including $j=u$). [\eqref{eq:u_polar} immediately implies that these projectors satisfy $P_i P_j = \delta_{ij} P_i$ for all $0 \le i,j \le N$, where $P_0 := P$.] The correction step is then: If the state is found in the codespace, do nothing; if it is in $\text{col}(P_i)$ for $1 \le i \le N$, apply $U_i^\dagger$. In other words:
\begin{equation}
\mathcal{R}(\rho) = P \rho P + 
\sum_{i = 1}^N U_i^\dagger P_i \rho P_i U_i.
\end{equation}

Having defined $\mathcal{R}$, we now compute $\mathcal{L}_\text{eff}$. It is convenient to define superoperators $\mathcal{H}$ and $\mathcal{D}_j$ such that the sensor's Liouvillian takes the form
\begin{equation}
\mathcal{L}(\rho) 
=
-i 
\omega \underbrace{[H, \rho]
\vphantom{\Big(\frac{1}{2} \{ L_j^2, \rho \}\Big)}}_{\mathcal{H}(\rho)}
+
\frac{1}{2T_2}\sum_{j=1}^N
\underbrace{
\Big(
L_j \rho L_j
- \frac{1}{2} 
\{ L_j^2, \rho \}
\Big)
}_{\mathcal{D}_j (\rho)} .
\end{equation}
We begin by computing $\mathcal{R H P}$. The fact that $P_j$'s are mutually orthogonal (for $0 \le j \le N$) immediately implies that
\begin{align}
\mathcal{R H P}(\rho)
=
[\mathcal{P}(H), \, \mathcal{P}(\rho)]
&=
 \frac{\gamma (\gv v_u \cdot \frakh )}{2} [Z_\textsc{l}, \mathcal{P}(\rho)] \nonumber,
\end{align}
using \eqref{eq:KL1}. For a logical $\rho = \mathcal{P}(\rho)$ may therefore write $\mathcal{R H P}(\rho) = [H_\text{eff}, \rho]$, where 
\begin{equation}
H_\text{eff} = \frac{\gamma (\gv v_u \cdot \frakh )}{2} Z_\textsc{l},
\end{equation}
as in \appref{app:dephasing_KL}.

We now turn to $\mathcal{R D}_j \mathcal{P}$ for $j \neq u$, which we expect to vanish since the code corrects the associated error $L_j$ by design. Assuming a logical $\rho$ to simplify the notation, we have
\begin{equation}
\mathcal{R D}_j \mathcal{P}(\rho)
=
\underbrace{P L_j \rho L_j P
\vphantom{\sum_{i=1}^N} 
}_\text{(i)}
- 
\underbrace{\frac{1}{2} P \{ L_j^2, \rho \} P
\vphantom{\sum_{i=1}^N} 
}_\text{(ii)}
+
\underbrace{\sum_{i=1}^N U_i^\dagger P_i L_j \rho L_j P_i U_i}_\text{(iii)}
\end{equation}
where the other terms vanish because the $P_j$'s are all orthogonal. Term (i) vanishes since $PL_j P=0$ for $j\neq u$, while term (ii) equals $m_{jj} \, \rho$. To evaluate term (iii), we first note that $U_u^\dagger P_u L_j P=0$. The equality holds trivially when \eqref{eq:u_polar} vanishes, and when it does not we have
\begin{align}
U_u^\dagger P_u L_j P
=
(m_{uu} - \lambda_u \gamma^2 - \alpha)^{-1/2} \; P L_u P_\textsc{r}
\sum_{i \in \mathcal{I}} \frac{w_{ji}}{\sqrt{d_i}} P_i U_i 
= 0
\end{align}
since $P_\textsc{r} P_i = 0$. This leaves
\begin{align}
\text{term (iii)}
&=
\sum_{i \in \mathcal{I}}
U_i^\dagger P_i L_j P\rho P L_j P_i U_i\\
&=
\sum_{i \in \mathcal{I}} \frac{1}{|d_i|} P F_i L_j P \rho P L_j F_i P \nonumber \\
&=
\sum_{i \in \mathcal{I}} \sum_{k, \ell \neq u} \frac{1}{|d_i|} w_{jk} w_{j\ell} \delta_{ik} \delta_{i \ell} d_k d_\ell  \rho   \nonumber \\
&= m_{jj} \, \rho. \nonumber
\end{align}
We therefore have the expected result $\mathcal{RD}_j \mathcal{P}=0.$

Finally, we turn our attention to $\mathcal{R D}_u \mathcal{P}$, which we do not expect to vanish (unless we are in the HNLS limit where $L_u = 0$), since we have designed our code such that $L_u$ is uncorrectable. As before, we have
\begin{equation}
\mathcal{R D}_u \mathcal{P}(\rho)
=
\underbrace{P L_u \rho L_u P
\vphantom{\sum_{i=1}^N} 
}_\text{(I)} 
- 
\underbrace{\frac{1}{2} P \{ L_u^2, \rho \} P
\vphantom{\sum_{i=1}^N} 
}_\text{(II)}
+
\underbrace{\sum_{i=1}^N U_i^\dagger P_i L_u \rho L_u P_i U_i}_\text{(III)}.
\end{equation}
From \eqref{eq:KL1}, the first term simplifies to
\begin{equation}
\text{term (I)} = \lambda_u \gamma^2 \, Z_\textsc{l} \rho  Z_\textsc{l},
\end{equation}
while term (II) becomes $m_{uu} \, \rho$. Treating the $i\neq u$ and $i=u$ parts of term (III) separately, we have
\begin{align}
\sum_{i \in \mathcal{I}} U_i^\dagger P_i L_u P L_u P_i U_i
=
\sum_{i \in \mathcal{I}} \frac{1}{|d_i|} P L_i L_u P \rho P L_u L_i P 
= \alpha \rho,
\end{align}
so
\begin{equation}
\mathcal{R D}_u \mathcal{P}(\rho)
=
\lambda_u \gamma^2 Z_\textsc{l} \rho Z_\textsc{l}
+
(\alpha - m_{uu}) \rho + U_u^\dagger P_u L_u \rho L_u P_u^\dagger U_u.
\label{eq:missing_lindblad}
\end{equation}
One immediately sees from \eqref{eq:missing_lindblad} that were it not for the $P_u$ measurement and feedback that we have added to $\mathcal{R}$ (the last term in the above equation), $\mathcal{L}_\text{eff}$ would not be of Lindblad form over the codespace, due to leakage into $\text{col}(P_\textsc{r})$. This is manifest through the mismatch between the $\lambda_u \gamma^2$ and $(\alpha - m_{uu})$ coefficients [cf.\ \eqref{eq:lindblad}]. By design, however, we have
\begin{align}
    U_u^\dagger P_u L_u P
   &=
    P U_u^\dagger L_u P\\
    &=
    (m_{uu} - \lambda_u \gamma^2 - \alpha)^{-1/2} \; P L_u P_\textsc{r} L_u P \nonumber \\
    &=
    \sqrt{m_{uu} - \lambda_u \gamma^2 - \alpha} \; P \nonumber
\end{align}
from \eqref{eq:u_polar}, which cancels the mismatched $\alpha-m_{uu}$ term in \eqref{eq:missing_lindblad}, giving the valid, trace-preserving Lindblad dissipator
\begin{equation}
\mathcal{R D}_u \mathcal{P}(\rho)
=
\lambda_u \gamma^2 \big( Z_\textsc{l} \rho Z_\textsc{l} - \rho \big).
\end{equation}

In summary, then, the effective logical dynamics in the limit of frequent recoveries is generated by
\begin{equation}
\mathcal{L}_\text{eff}(\rho)
=
-i[\omega H_\text{eff}, \rho] 
+ L_\text{eff} \, \rho  L_\text{eff}^\dagger 
-\frac{1}{2} \{ L_\text{eff}^\dagger L_\text{eff}, \rho \},
\end{equation}
where $H_\text{eff} = \gamma (\gv v_u \cdot \frakh ) Z_\textsc{l}/2$ and $L_\text{eff} = \gamma \sqrt{\frac{\lambda_u}{2 T_2}} Z_\textsc{l}$, with $Z_\textsc{l} = \ket{0_\textsc{l}} \!\! \bra{0_\textsc{l}} - \ket{1_\textsc{l}} \!\! \bra{1_\textsc{l}}$. In other words, at the logical level, the sensor accumulates a phase at a rate set by $\gamma$ and the overlap of $L_u$ with $H$, while also losing phase coherence at a rate set by $\gamma$ and $\lambda_u$. \eqref{eq:eta_QEC} follows immediately from this result.

\section{\label{app:dephasing_ex}Dephasing code example}

Consider a sensor comprising $N \ge 3$ identical probing qubits arranged in a ring, and equally spaced from one another. Suppose that the noise correlation coefficient between two qubits depends only on the distance between them, so that neighboring qubits have a coefficient of $\alpha_1$, next-to-nearest neighbors have $\alpha_2$, and so on up to $\alpha_\text{max}$ for the most distant qubits, as shown in \figref{fig:ring}. (N.b., the special case of $\alpha_i=0$ describes a lack of correlation.) We emphasize that in practice, the distance between two qubits is a poor predictor of how strongly correlated the noise in their gaps is. As discussed in the main text, other factors, like proximity and relative orientations to nearby fluctuators, for instance, are often more important. Accordingly, this example is not meant to provide a particularly realistic model, but rather, an illustrative one which can be solved exactly, and which has normal noise modes with a simple physical interpretation.

\begin{figure}[htbp] 
\includegraphics[width=0.9\textwidth]{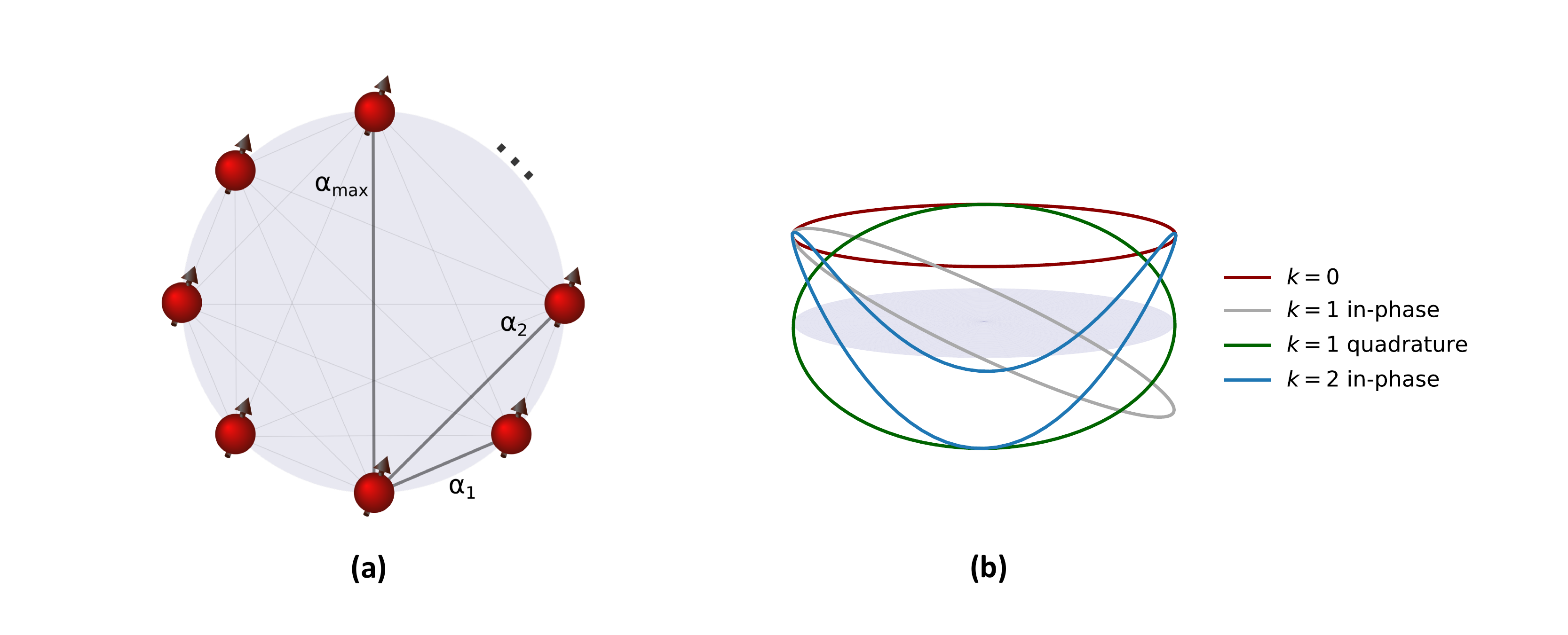}
\caption{(a) A ring of $N\ge3$ evenly-spaced probing qubits. The lines connecting each pair denote an arbitrary correlation strength $\alpha_i$ in the energy gap noise felt by either qubit. (b) The first few Fourier modes in the sensor resulting from the spatial noise correlations. (Qubits not shown.) The height of each point along a curve denotes the weight of $Z_j$ for qubit $j$ in the corresponding Lindblad error operator, for modes $\gv v_0$, $\gv v_{1,0}$, $\gv v_{1,\pi/2}$, $\gv v_{2,0}$.}
\label{fig:ring}
\end{figure}

The qubits in this sensor are assumed to undergo Markovian dephasing described by \eqref{eq:dephasing_lindblad}, with $\mathfrak h_j=1$ for simplicity and 
\begin{equation}
    C
    =
    \begin{pmatrix}
    1 & \alpha_1 & \alpha_2 &  & \alpha_2 & \alpha_1 \\
     \alpha_1 & 1 & \alpha_1 & \cdots & \alpha_3 & \alpha_2 \\
    \alpha_2 & \alpha_1 & 1 & & \alpha_4 & \alpha_3 \\
    & \vdots & & \ddots  & & \vdots \\
    \alpha_2 & \alpha_3 & \alpha_4 & & 1 & \alpha_1 \\
    \alpha_1 & \alpha_2 & \alpha_3 & \cdots & \alpha_1 & 1
    \end{pmatrix},
\end{equation}
where we assign qubit numbers/labels sequentially based on their location. Notice that $C$ is a circulant matrix, so it is diagonalized by a discrete Fourier transform matrix \cite{davis1979circulant}. This means that the eigenvectors of $C$ are simply the spatial Fourier modes on the ring. These are often chosen to be complex vectors of the form
\begin{equation}
\gv w_k = \frac{1}{\sqrt{N}} (1, \omega_k, \omega_k^2, \dots, \omega_k^{N-1})^\top,
\end{equation}
where $\omega_k:=\exp(2\pi i k/N)$ for $k=0,\dots,N-1$. The corresponding eigenvalues are 
\begin{equation}
\lambda_k = 1+\alpha_1 \omega_k + \alpha_2 \omega_k^2 + \dots
+ \alpha_2 \omega_k^{N-2} + \alpha_1 \omega_k^{N-1} 
\in \mathbb{R}.
\end{equation}
Since these eigenvalues come in degenerate pairs ($\lambda_k = \lambda_{N-k}$ for $k \ge 1$), we can equivalently form a real eigenbasis for $C$ from $\gv w_k \pm \gv w_{N-k}$, in keeping with the convention from the main text of using $\gv v_k \in \mathbb{R}^N$. The first such eigenvector is the $k=0$, or constant Fourier mode
\begin{equation}
    \gv v_0 = \frac{1}{\sqrt{N}} (1,\dots,1)^\top = \frakh /\sqrt{N},
\end{equation}
Higher wavenumbers $k \ge 1$ each describe a pair of eigenvectors with a $\pi/2$ phase offset on the ring:
\begin{align}
\gv v_{k,0}
&=
\big(1, \, \cos( \theta_k), \, \cos( 2\theta_k),
\dots, \,
\cos[(N-1)\theta_k] \big)^\top\\
\gv v_{k,\pi/2}
&=
\big(0, \,\sin( \theta_k), \,\sin( 2\theta_k),
\dots, \,
\sin[(N-1)\theta_k] \big)^\top, \nonumber
\end{align}
where $\theta_k = 2\pi k/N$. The first few of these Fourier modes are illustrated in \figref{fig:ring}. The key observation here is that all but the $k=0$ noise mode are orthogonal to $\frakh$. This forces us to take $\gv v_u = \gv v_0$ as the mode to leave uncorrected. Any other choice of $u$ would give $\eta_\text{QEC}^{(u)} = \infty$, and is therefore inadmissible. The eigenvalue of $C$ associated with this mode is given by
\begin{equation}
    \lambda_u = \gv v_u^\top C \gv v_u =  \frac{1}{N} \, \frakh ^\top C \frakh.
\end{equation}
From \eqref{eq:eta_QEC}, this gives an achievable sensitivity of 
\begin{equation}
\eta_\text{QEC} = \frac{\sqrt{\frakh^\top C \ \frakh }}{N} \sqrt{\frac{2e}{T_2}} 
=
\sqrt{1 + \alpha_1 + \alpha_2 + \dots + \alpha_\text{max} } \sqrt{\frac{2e}{N T_2}}.
\end{equation}
Notice that the scaling of $\eta_\text{QEC}$ with $N$ here has two components: (i) a generic $1/\sqrt{N}$ improvement as one adds qubits, and (ii) a prefactor that depends on how $\sum_i \alpha_i$ changes with $N$.

Comparing with $\eta_\text{par} = \sqrt{2e /N T_2}$, one sees that QEC provides an advantage over a parallel sensing scheme (i.e., one with initial state $\ket{\psi_0} = \ket{+}^{\otimes N}$) when $\alpha_1 + \dots +\alpha_\text{max} < 0$. Therefore, in this particular setting, there must be negative correlations in the noise (i.e., certain $\alpha_j < 0$) for QEC to provide an advantage over parallel sensing. This is not required in general however, see e.g., the Supplementary Information for Ref.~\cite{layden2018spatial} for a counterexample.

Finally, we compare the sensitivity offered by QEC in this model with that offered by a GHZ scheme which uses an initial state $\ket{\psi_0} = \frac{1}{\sqrt{2}} (\ket{0\dots 0} + \ket{1\dots 1} )$. \eqref{eq:eta_GHZ} of the main text immediately gives
\begin{equation}
\eta_\text{GHZ} = \frac{\sqrt{\frakh ^\top C \frakh }}{N} \sqrt{\frac{2e}{T_2}}  = \eta_\text{QEC}.
\end{equation}
Therefore, with this particular noise model, our QEC scheme offers the same sensitivity as GHZ sensing (in the $\Delta t \rightarrow 0$ limit) but does not surpass it. There is a simple explanation for this apparent coincidence: taking $\gv v_u = \gv v_0$ and $\gamma = \gamma_\text{max} = \sqrt{N}$ gives logical states $\ket{0_\textsc{l}} = \ket{0\dots 0}$ and $\ket{1_\textsc{l}} = \ket{1\dots 1}$, resulting in the same initial state as the GHZ scheme. Moreover, a simple calculation shows that both of these logical states are in the kernel of every Lindblad error with $k\ge 1$. In other words, span$ \{ \ket{0\dots 0}, \ket{1\dots 1} \}$ is a decoherence-free subspace of $\{ L_k \}_{k \ge 1}$, and so the recovery $\mathcal{R}$ reduces to the identity channel (i.e., doing nothing) \cite{layden2018spatial}. This means that for this model, GHZ sensing is a special case of our error-correction scheme, for a particular choice of the adjustable parameter $\gamma$. Since $\eta_\text{QEC}$ is independent of $\gamma$, it follows that $\eta_\text{QEC} = \eta_\text{GHZ}$ here. Of course, this is particular to the present model, and is not the case in general, even with purely positive noise correlations \cite{layden2018spatial}.

\end{document}